\documentclass{article}[10pt]
\usepackage[a4paper]{geometry}

\usepackage{amsmath}
\usepackage{amsthm}
\usepackage{thmtools} 
\usepackage{thm-restate}
\usepackage{my_theoremshal}
\usepackage{hyperref}
\usepackage{amssymb}
\usepackage{url}
\usepackage{comment}
\usepackage{macro}
\usepackage{xspace}
\usepackage{stmaryrd}
\usepackage{mathtools}
\usepackage{xcolor}
\usepackage{relsize}
\usepackage{todonotes}

\usepackage{algorithm}
\usepackage[noend]{algpseudocode}
\usepackage{tikz}
\usetikzlibrary{arrows,automata}
\usetikzlibrary{positioning}
\usepackage{caption}
\usepackage{subcaption}

\title{Complexity of Membership and Non-Emptiness Problems in Unbounded Memory Automata} 

\author{Clément Bertrand
\and
Cinzia {Di Giusto}
\and
Hanna Klaudel
\and
Damien Regnault}
\date{}

\begin{document}

\maketitle

\begin{abstract}
  We study the complexity relationship between three models of unbounded memory automata: $nu$-automata (\nuauto), Layered Memory Automata (\lama)and History-Register Automata (\hra). These are all extensions of finite state automata with unbounded memory over infinite alphabets. We prove that the membership problem is NP-complete for all of them, while they fall into different classes for what concerns non-emptiness. The problem of non-emptiness is known to be Ackermann-complete for \hra, we prove that it is PSPACE-complete for \nuauto.
\end{abstract}

\section{Introduction}

We study unbounded memory automata for words over an infinite alphabet, as introduced in \cite{Kaminski1994, Neven04}.
Such automata model essentially dynamic generative behaviours, \ie programs that generate an unbounded number of distinct resources each with its own unique identifier (e.g. thread creation in Java, XML). For a detailed survey, we refer the reader to \cite{BertrandKP22,Kara16}. 
We focus, in particular, on three formalisms, $\nu$-automata (\nuauto) \cite{DBLP:conf/acsd/DeharbeP14,BertrandPKL18}, 
Layered Memory Automata (\lama) \cite{BertrandKP22} and \hra for History-Register Automata \cite{GT2016}. 
All these models are extensions of finite state automata with a memory storing letters. 
The memory for \hra is composed of registers (that can store only one letter) and histories (that can store an unbounded number of letters). Among the distinctive features of \hra we have the ability to reset registers and histories, and to select, remove and transfer individual letters. Whereas 
the memory for the other two models consists of a finite set of variables. 
In those formalisms, variables must satisfy an additional constraint, referred to as injectivity, meaning that they cannot store shared letters. 
Moreover, variables can be emptied (reset) but single letters cannot be removed. We know that \lama are more expressive than \nuauto as the former are closed under intersection while it is not the case for the latter ones. 

We tackle two problems: membership and non-emptiness. From a practical point of view,  automata over infinite alphabets can be used to identify patterns in link-stream analysis \cite{BertrandPKL18}. In such a scenario, the alphabet is not known in advance (open systems) and runtime verification can help to recognise possible attacks on a network by looking for specific patterns. This problem corresponds to checking whether a pattern (word) belongs to a language (the membership problem). Concerning non-emptiness, this is the “standard” problem to address while considering automata in general.

Fig.  \ref{fig:graphe_expre} depicts the unbounded memory automata known in the literature (to the best of our knowledge). An implementation exists for \nuauto  and \lama which includes an implementation of the membership algorithm, but we have not found anything neither for Data Automata (DA) nor Class Memory Automata (CMA). Both DA and CMA are incomparable classes wrt to \hra, hence we chose not to consider them. Fresh-Register automata (FRA), nu-automata, LaMA, HRA are instead related from the expressiveness point of view. Given the similarities between those formalisms, the existence of implementations and the lack of complexity results we find it natural to consider these classes of automata.

Application-wise, \nuauto, called resource graphs in \cite{DBLP:conf/acsd/DeharbeP14}, model the use of unbounded resources in the $\pi$-calculus, aiming at minimizing them. Then, runtime verification on link-streams was the initial motivation for the introduction of (timed)\nuauto \cite{BertrandPKL18}. In subsequent works, LaMA have been introduced to be able to construct the synchronous product, hence being able to express the synchronization of resources. This entails the closure by intersection, which is interesting when one wants to define a language of expressions, an extension of (timed) regular expressions, which was proposed in the PhD thesis of Clément Bertrand \cite{bertrand:tel-03172600}.

\begin{figure}
    \centering
    \begin{tikzpicture}[xscale=.9, yscale=1.2]
\node at (-1,3.0) {\large \underline{Bounded memory}};
\node at (6,3.0) {\large \underline{Unbounded memory}};

\node [draw, rounded corners , thick] (FMA) at (-2.3,2) {FMA\cite{Kaminski1994}};
\node [draw, rounded corners , thick] (VFA) at (-2.3,0) {VFA \cite{Grumberg10}};
\node [draw, rounded corners , thick] (FRA) at (2.8,2) {FRA \cite{Tzevelekos11}};
\node [draw, rounded corners , thick] (GRA) at (0.2,0) {GRA \cite{KaminskiZ10}};
\node [draw, rounded corners , thick,fill=yellow] (nu) at (2.6,1) {\nuauto \cite{BertrandPKL18}};
\node [draw, rounded corners ,  thick, fill=yellow] (ACM) at (5.2,1) {\lama \cite{BertrandKP22}};
\node [draw, rounded corners , thick, fill=yellow] (HRA) at (8.3,1) {\hra \cite{GT2016}};
\node [draw, rounded corners , thick] (DA) at (11.2,1) {\begin{tabular}{c}DA\cite{BojanczykDMSS11}\\CMA \cite{BJORKLUND2010702}\end{tabular}};

\draw [double, very thick] (1.5,-1.1)--(1.5,2.7);

\draw [->, thick] (FMA)--node[above]{\cite{Tzevelekos11}}(FRA);
\draw [->, thick] (VFA)--node[above]{\cite{Grumberg10}}(GRA);
\draw [->, thick] (FMA)--node[above]{\cite{KaminskiZ10}}(GRA);
\draw [->,  thick] (GRA)--node[below]{\cite{BertrandKP22}}(ACM);
\draw [->,  thick] (FRA)--node[right]{\cite{BertrandKP22}}(ACM);
\draw [->, thick] (FRA)--node[above]{\cite{GT2016}}(HRA);
\draw [->,  thick] (ACM)--node[below]{\cite{BertrandKP22}}(HRA);
\draw [->, thick] (FMA)--node[above]{\cite{BertrandPKL18}}(nu);
\draw [->,  thick] (nu)--node[above]{\cite{BertrandKP22}}(ACM);
\draw [thick, dotted] (FMA)--node[left]{\cite{Grumberg10}}(VFA);
\draw [thick, dotted] (HRA)--node[above]{\cite{GT2016}}(DA);
\draw [->, thick] (FRA) to[out = 0, in = 160]node[above]{\cite{GT2016}}(DA);
\draw [->, thick] (VFA) to[out=-10, in = 195]node[above]{\cite{Grumberg10}} (DA);
    \end{tikzpicture}
    \caption{A classification of memory automata 
    from \cite{BertrandKP22}. Arrows represent strict language inclusion, while dotted lines denote language incomparability. Formalisms studied here are in yellow.}
    \label{fig:graphe_expre}
\end{figure}
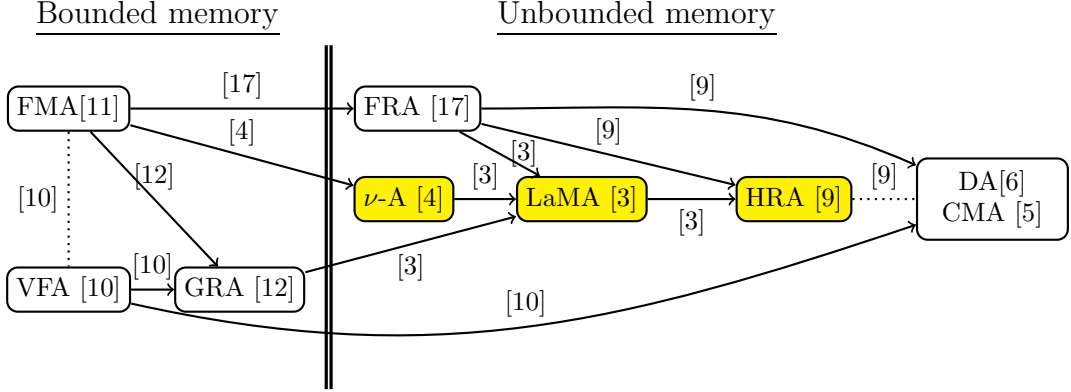
For a precise discussion on the relations among these formalisms see \cite{BertrandKP22}, here we just recall the hierarchy: cfr. Fig. ~\ref{fig:graphe_expre}. 
Apart from expressiveness, a number of questions concerning complexity remains open. We know  that checking whether the language recognized by an \hra is empty  or not (referred to as the non-emptiness problem) is Ackermann-complete \cite{GT2016}. But the question has not been addressed for \nuauto and \lama. 
For finite-memory automata (FMA), it is known that membership (testing whether a word belongs to the language) and non-emptiness are NP-complete \cite{SAKAMOTO2000297}. Knowing whether a language is included in another is undecidable for FMA when considering their non-deterministic version, but it is PSPACE-complete for deterministic ones \cite{MurawskiRT18}.
In \cite{Grumberg10} it is shown that the non-emptiness problem for Variable Finite Automata (VFA) is NL-complete, while membership is NP-complete. For FRA and Guessing-Register Automata (as they are called in \cite{Kara16}) we only know that both problems are decidable but we do not know the accurate complexity class. 
Finally, for data-languages where data-words are sequences of pairs of a letter from a finite alphabet and an element from an
infinite set and the latter can only be compared for equality, the non-emptiness problem for FMA is PSPACE-complete \cite{Demri09}, for DA and CMA, membership and non-emptiness  are only  shown to be decidable, but no complexity is given \cite{BojanczykDMSS11,BJORKLUND2010702}.

\subparagraph{Contributions.}
In this paper, we close some open problems on the complexity of  membership  and non-emptiness.
We first prove that testing membership for \hra, \nuauto and \lama is an equivalent problem. Then we address complexity and  show that the problem is NP-complete with a reduction of 3-SAT to \lama.
Non-emptiness appears to be a much harder problem. 
We show that the non-emptiness problem is PSPACE-complete for \nuauto by reducing   TQBF (True Fully Quantified Boolean Formula) to \nuauto. 

The paper is organized as follows. The three formalisms are introduced and the main differences recalled quickly in Section~\ref{sec:formalism}. Section~\ref{sec:member} presents the complexity of the membership problem and Section~\ref{sec:empty} the non-emptiness one. Finally, Section~\ref{sec:concl} concludes with some remarks.

\section{Formalisms}\label{sec:formalism}

Our formalisms are all generalizations of finite state automata with memory over an infinite alphabet $\alphab$. For all of them,  configurations $(q,M)$ are  pairs of a state of the automaton plus a memory context. A memory context assigns a set of letters  to each identifier of the memory, variable or history depending on the formalism under consideration. 

\begin{definition}[Memory context] \label{def:contexte-memoire}
Given a finite set of memory identifiers or variables $V$
and an infinite alphabet $\alphab$, we define a memory context $M$ as an assignment: $M: V \rightarrow 2^\alphab$
where $M(v) \subset \alphab$ is the finite set of letters \emph{assigned} to $v$. 
\end{definition}


The definition of accepted language common to the  three formalisms is, as customary:
\begin{definition}[Accepted language]
For an automaton $A$ (\lama, \nuauto or \hra), the language of $A$ is the set of words recognized by $A$: $\lang{A}=\{ w\in \alphab^* \mid  (q_0,M_0) \xRightarrow[A]{w} (q_f,M) \mbox{ s.t. } q_f\in F\}$, where $\xRightarrow[A]{w}$ is the extension to sequences of transitions of $\xrightarrow[A]{u}$.
\end{definition}

\subsection{n-Layered Memory Automata}

We start with $n$-Layered Memory Automata ($n$-\lama). The idea is that finite state automata are enriched with $n$ layers each containing a finite number of variables.  Variables on the same layer satisfy the   
\emph{injectivity} constraint:  variables on a given layer $l \in [1,n]$  (denoted $v^l$) do not share letters of the alphabet:  $\forall v_1 \neq v_2 \in V, M(v_1^l) \cap M(v_2^l) = \emptyset$. 
Upon reading a letter, a transition can test if the letter is already stored in a set of variables and add a letter to a set of variables. The non-observable transition ($\varepsilon$-transition) empties a set of variables.

\begin{definition}[$n$-\lama] An $n$-\lama $A$ is defined by the tuple $(Q, q_0, F, \Delta, V, n, M_0)$, where:
\begin{itemize}
\item $Q$ is a finite set of states,
\item $q_0 \in Q$ and $F \subseteq Q$ are respectively an initial state and a set of final states,
\item $\Delta$ is a finite set of transitions,
\item $n$ is the number of layers, and $V$ is a finite set of variables, denoted $v^l$ with $l\in [1,n]$,
\item $M_0: V \mapsto 2^\alphab$ is an initial memory context. 
\end{itemize}
\end{definition}

A \emph{fresh letter at layer $l$}, is a letter that is associated with no variable of this layer. 
The set of transitions $\Delta = \setobs \cup \setsilent$ encompasses two kinds of transitions with $\setobs$  the set of \emph{observable transitions} that consume a letter of the input
and $\setsilent$  the set of \emph{non-observable transitions}, that  do not consume any letter of the input but can reset a set of variables.

\begin{definition}[Observable transition]\label{def:transitionlama}
An observable transition is a tuple of the form:
$\quad \delta = (q, \alpha, q') \in \setobs$
where:
\begin{itemize}
\item $q,q' \in Q$ are the source and destination states of the transition,
\item $\alpha: [1,n] \rightarrow (V \times  \{\reads,\writes\})\cup \{ \sharp \} $, such that $\alpha(l) = (v^l, \mathbf{x})$  for $\mathbf{x} \in  \{\reads,\writes\} $ and  for some $v^l \in V$ indicates for each layer $l$ which variable is examined by the transition.

\end{itemize}
\end{definition}

Notice that only one variable per layer can be examined, and it is not possible to have $\alpha(l) = (v^k, \mathbf{x})$ with $l \neq k$.
The special symbol $\sharp$ indicates that no variable is to be read or written for a specific layer.

\begin{remark}[Any-letter transition]
If  $\forall l \in [1,n], \alpha(l) = \sharp$ (i.e., no variable is examined) then the transition is executed consuming whatever letter is in input.
\end{remark}

\begin{definition}[Non-observable transition] \label{def:eps_transitionlama}
A non-observable transition is a tuple of the form $\delta_\varepsilon = (q,\rset,q') \in \setsilent$ where:
\begin{itemize}
\item $q,q' \in Q$ are the source and destination states of $\delta_\varepsilon$,
\item $\rset \subseteq 2^{V}$ is the set of variables reset (i.e., emptied) by the transition.
\end{itemize}
\end{definition}

\begin{definition}
The semantics of an $n$-\lama $A=(Q, q_0, F, \Delta, V, n, M_0)$ is defined as:
\begin{itemize}
    \item An observable transition $(q, \alpha, q')$ can be executed on an input letter $u$ from memory context $M$ leading to $M'$: $(q,M) \xrightarrow[A]{u} (q',M')$ if for each  $\alpha(l) \neq \sharp$ such that $\alpha(l) = (v^l, \mathbf{x})$:
        \begin{itemize}
        \item if $\mathbf{x}= \reads$, then $u \in M(v^l)$ and  $M'(v^l) = M(v^l)$ ; 
        \item if $\mathbf{x}= \writes$, then $u$ is fresh for layer $l$ 
        and $u$ is added to  $v^l$ in the reached memory context: $M'(v^l) = M(v^l)\cup\{u\}$. 
        \end{itemize}
    All variables $v^l$ not labeled through $\alpha$ remains associated to the same letters : if $\alpha(l) = \sharp$ or $\alpha(l) = (v_1^l,x)$ and $v_1^l \neq v^l$ then $M'(v^l) = M(v^l)$. 
    \item A non-observable transition $(q,\rset,q')$ can be executed from memory context $M$ without reading any input letter leading to $M'$: $(q,M) \xrightarrow[A]{\varepsilon} (q',M')$, where $\forall v^l\in \rset: M'(v^l)=\emptyset$ and otherwise $M'(v^l)=M(v^l)$.
\end{itemize}
\end{definition}

\begin{example}\label{ex:lama}
    Let $\plang{p} =\{ u_1\dots u_s \mid \forall j, k>0, j\neq k, \ u_{j\cdot p} \neq u_{k\cdot p}  \}$, be the language recognizing words where  the letters at positions, which are multiples of $p$ are all different whereas the others are not constrained.  Fig.~\ref{fig:lama_Ap} depicts a $2$-\lama for  $\plang{2} \cap \plang{3}$.
\end{example}
\begin{figure}
    \centering
        \begin{tikzpicture}[xscale=2.5, yscale=1]
\node [draw, circle, double] (q1) at (1,0)  {$q_1$};
\draw [->](0.7,0)--(q1);
\node [draw, circle, double] (q2) at (2,0) {$q_2$};
\node [draw, circle, double] (q3) at (3,0) {$q_3$};
\node [draw, circle, double] (q4) at (4,0) {$q_4$};
\node [draw, circle, double] (q5) at (5,0) {$q_5$};
\node [draw, circle, double] (q6) at (6,0) {$q_6$};

\draw [->, thick] (q1)--node[above]{$\sharp$}(q2);
\draw [->, thick] (q2)--node[above]{$(X^1,\writes)$}(q3);
\draw [->, thick] (q3)--node[above]{$(Y^2,\writes)$}(q4);
\draw [->, thick] (q4)--node[above]{$(X^1,\writes)$}(q5);
\draw [->, thick] (q5)--node[above]{$\sharp$}(q6);
\draw [->, thick, rounded corners] (q6)--(6,-1)--node[above]{$(X^1,\writes), (Y^2,\writes)$}(1,-1)--(q1);

    \end{tikzpicture}

    \caption{ A $2$-\lama $A_{p}$
    recognizing 
    $\plang{2} \cap \plang{3}$ from Example \ref{ex:lama} 
    }\label{fig:lama_Ap}
   
\end{figure}
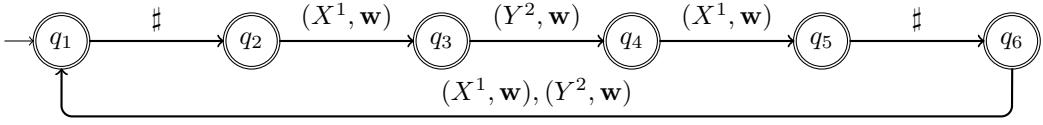

\subsection{{$\nu$}{nu}-automata}

 $\nu$-automata (\nuauto) can be seen as a restricted version of \lama with only one layer. 
Hence, each variable is constrained under  the  injectivity property, and no letter can be stored in more than one variable.

\begin{definition}[\nuauto]
A \nuauto is defined as a tuple $(Q, q_0, F, \Delta, V, M_0)$, where 
\begin{itemize}
\item $Q$ is a finite set of states containing an initial state $q_0 \in Q$ and a set of final ones $F \subseteq Q$,
\item $V$ is a finite set of variables that may initially be storing a finite amount of letters from the infinite alphabet $\alphab$, as specified by the initial memory context $M_0$, 
\item and $\Delta$ is a finite set of transitions.
\end{itemize}
\end{definition}

As before,  $\Delta = \setobs \cup \setsilent$ where  $\setobs$ is  the set of \emph{observable transitions} and $\setsilent$  is the set of \emph{non-observable} ones. 
Differently from \lama, observable transitions are decoupled in read and write transitions.

\begin{definition}[Observable transition]\label{def:transitionnu}
An observable transition can be of two kinds: 
 $(q, v, \reads, q')$ and $(q,  v,\writes, q')$ ($\reads$ for read and $\writes$ for write) where $q,q' \in Q$ are the source and destination states and  $v \in V$.
\end{definition}

\begin{definition}[Non-observable transition] \label{def:eps_transitionnu}
A non-observable transition is a tuple of the form $\delta_\varepsilon = (q,\rset,q') \in \setsilent$ where:
$q,q' \in Q$ are the source and destination states of $\delta_\varepsilon$,
 $\rset \subseteq 2^{V}$ is the set of variables reset by the transition.
\end{definition}

\begin{definition}
The semantics of a \nuauto $A = (Q, q_0, F, \Delta, V, M_0)$ is defined as:
\begin{itemize}
    \item An observable transition $(q, v, \mathbf{x}, q')$ reading input letter $u$ can be executed from memory context $M$ leading to $M'$: $(q,M) \xrightarrow[A]{u} (q',M')$ if for each  $v$:
    \begin{itemize}
    \item if $\mathbf{x}= \reads$, then  $u \in M(v)$ and $M'(v) = M(v)$; 
    \item if $\mathbf{x}= \writes$, then $u$ is fresh in $M$ 
    and  $u$ is added to  $v$ in the reached memory context: $M'(v) = M(v)\cup\{u\}$.  
    \end{itemize}
    All other variables $v_1 \neq v$, remains associated to the same letters $M'(v_1) = M(v_1)$. 
    \item A non-observable transition $(q,\rset,q') \in \setsilent$ can be executed from memory context $M$ leading to $M'$: $(q,M) \xrightarrow[A]{\varepsilon} (q',M')$ without reading any input letter, where $\forall v\in \rset: M'(v)=\emptyset$ and otherwise $M'(v)=M(v)$.
\end{itemize}
\end{definition}

\begin{remark} 
Analogously to \lama, we consider \textit{any-letter transitions}, denoted by $(q,\sharp,q')$ with $\sharp \not\in \alphab$, which are enabled whenever a letter is read and the memory context of the target configuration is the same as the origin's one. 
\end{remark}

Notice that any-letter transitions do not alter the expressive power of \nuauto nor the complexity of its problems. Indeed, it is a sort of macro that can be encoded by  a set of transitions searching for the presence of a letter or its freshness over the whole set $V$. 
To do so, one needs as many reading transitions as variables to allow the firing with any letter in memory. For fresh letters, one needs a transition writing in an extra variable, which is reset immediately after. 

\subsection{History-Register Automata}

\hra are automata provided with a finite set $H$ of histories, i.e., variables storing a finite subset of letters of the infinite alphabet $\alphab$. 
To simplify the presentation, we consider \hra defined only with histories and no registers. The latter does not provide additional expressiveness \cite{GT2016}. An important distinction between  \hra and \lama or \nuauto is that different histories are allowed to store the same letter (\ie there is no injectivity constraint).
Thus,  an observable transition is annotated with the exact set of histories that should contain the letter read to enable it. This entails that for each observable transition the whole memory has to be explored while \lama allow ignoring some layers using symbol $\sharp$ (this can be crucial while  implementing the formalisms\footnote{Implemented in tool available at 
\url{https://github.com/clementber/MaTiNA/tree/master/LaMA}}).

\begin{definition}[\hra]
A History-Register Automata is defined as a tuple of the form $A = (Q, q_0, F, \Delta, H, M_0)$
where $Q$ is the set of states, $q_0$ the initial one, $F$ the set of final ones, $\Delta$ the set of transitions, $H$ a finite set of histories and $M_0$ the initial memory context.
The set of transitions $\Delta = \setobs \cup \setsilent$ are of the form:
\begin{itemize}
    \item  $(q,H_r,H_w,q')\in \setobs$ where $H_r, H_w \subseteq H$ (for read and write),  which is an observable transition and   
    \item $(q,H_\emptyset,q')\in \setsilent$ where $H_\emptyset \subseteq H$, which is a non-observable transition. 
\end{itemize}
\end{definition}

An observable transition $(q,H_r,H_w,q')$ is enabled if letter $u$ is  present in exactly all the histories in $H_r$ and not present in $H \setminus H_r$. After the transition, $u$ is present only in the histories in $H_w$. Notice that this allows moving an input letter from one set of histories to another, or even forgetting it if $H_w = \emptyset$. This is not possible in \nuauto and \lama. Finally, if $H_r = \emptyset$ then the input letter has to be fresh (absent from every history).

\begin{definition}
The semantics of an \hra $A = (Q, q_0, F, \Delta, H, M_0)$ is defined as:
\begin{itemize}
    \item an observable transition $(q,H_r,H_w,q')$ is enabled for memory context $M$ 
    when reading letter $u \in \alphab$: $(q,M) \xrightarrow[A]{u} (q',M')$ if $u\in M(h_r) \Leftrightarrow h_r\in H_r$
    and $\forall h_w\in H_w: M'(h_w)=M(h)\cup\{u\}$ and $\forall h\not\in H_w$: $M'(h)=M(h) \setminus \{u\}$;

    \item a non-observable transition $(q,H_\emptyset,q')$ is enabled for any memory context $M$ and allows to move from configuration $(q,M)$ to $(q',M')$: $(q,M) \xrightarrow[A]{\varepsilon} (q',M')$, where all the histories in $H_\emptyset$ have been reset in $M'$. 
\end{itemize}
\end{definition}

\begin{example}\label{ex:hra}
Fig. ~\ref{fig:HRA_example} depicts an HRA that, with an initially empty memory context, which recognizes the language
$$ \begin{array}{ll}\{ u_1 u_2 \dots u_n \mid & \exists k < n, \forall i,j \in [1,k], u_i =u_j \Leftrightarrow i = j, \\ & \forall m \in ]k,n], \exists p  < m, u_p = u_m, p\text{ mod } 2 \neq m\text{ mod } 2, \not\exists q \in ]p,m[, u_m = u_q \} \end{array}$$
The two transitions looping between states $q_{ow}$ and $q_{ew}$ allow us to recognize words where the first $k$ letters are all different from each other. Letters are stored in histories $O$ (odd) and $E$ (even) to remember the parity of the position they are read at.
The transitions between states $q_{er}$ and $q_{or}$ allow us to recognize  words whose suffix is only composed of repetitions of the previous $k$ letters, with the additional constraint that those letters can only occur at a position  with opposed parity wrt the previous occurrence. Thus, if a letter was read for the last time at an even position, it is stored in history $E$ and can only be read in an odd position. Once it is read, it is transferred to the $O$ history to remember it can only be read at an even position the next time.
\end{example}

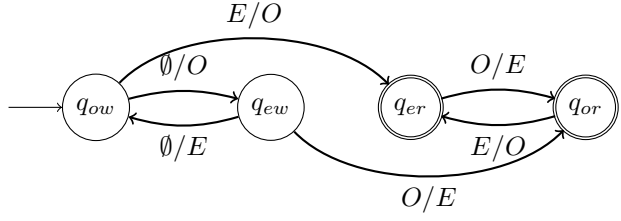
\begin{figure}
    \centering
    \begin{tikzpicture}[xscale=2.3, yscale=1.5]
\node [draw, circle](qow) at (0,0) {$q_{ow}$};
\draw [->](-0.5,0)--(qow);
\node [draw, circle](qew) at (1,0) {$q_{ew}$};
\node [draw, circle, double](qer) at (1.8,0) {$q_{er}$};
\node [draw, circle, double](qor) at (2.8,0) {$q_{or}$};
\draw [->, thick] (qow) to[out=18, in = 162]node[above]{$\emptyset/O$} (qew);
\draw [->, thick] (qew) to[out = 200, in = -20]node[below]{$\emptyset/E$} (qow);
\draw [->, thick] (qow) to[out = 55, in = 125]node[above]{$E/O$} (qer);
\draw [->, thick] (qew) to[out = -55, in = -125]node[below]{$O/E$} (qor);
\draw [->, thick] (qor) to[out = 200, in = -20]node[below]{$E/O$} (qer);
\draw [->, thick] (qer) to[out=20, in = 160]node[above]{$O/E$} (qor);

    \end{tikzpicture} 
    \caption{Example of an \hra $A_h$.} 
    \label{fig:HRA_example}
\end{figure}

\section {Complexity of the membership problem}\label{sec:member}

We  know that each \nuauto can be encoded into a \lama and respectively each \lama can be encoded into an \hra both recognizing  the same language \cite{BertrandKP22}.
The encoding from \lama to \hra is  exponential in the number of layers, hence we know that the complexity of problems for \hra gives an upper bound to the complexity of the same problem for \lama and \nuauto.
In this section, we show that the complexity of the membership problem (i.e.,  given an automaton $A$ and a word $w$  decide whether $w \in \lang{A}$), for these three automata models, falls in the same class. 
To do so, we show that the membership problem for \lama can be simulated using  \nuauto, and the same can be done for \hra   using \lama.

\subparagraph*{Simulating the membership for \lama in {$\nu$}{nu}-automata}
The idea is to represent an n-\lama  as a product of $n$ \nuauto, one for each layer. The main limitation is that having just one layer makes  the injectivity constraint stronger. Indeed, it is not possible to trivially treat a same letter that is stored on different layers.
To cope with this difficulty, we rename the word under consideration, replacing consistently each letter with a sequence 
of new ones - one per layer of the \lama: i.e., for an $n$-\lama the letter $u\in w$ is replaced by the letters $u^1, ..., u^n$ where all the $u^i$ are different in order to have the letters belonging to different layers all distinct from each other.
This renaming is always possible as the alphabet $\alphab$ is infinite. For example, for the word $aba$, a consistent renaming, for a 2-LaMA, could produce 
$ a^1\,a^2\,b^1\,b^2\,a^1\,a^2$.

\begin{definition}[Renaming]
 $\xi^n : \alphab \rightarrow \alphab^n$ is a \emph{renaming} function that given a letter $u \in \alphab$ generates a new sequence of $n$ letters $u^1\dots u^n$ with for all $i\neq j \in [1,n]$ $u^i \neq u^j$ and such that if $u_1 \neq u_2$ then for all $i, j \in [1,n]$, $u_1^i \neq u_2^j$. 
$\xi^n(u_1\dots u_m) = u_1^1\dots u_1^n  \dots u_m^1 \dots u_m^n$ is  its pointwise extension to words .    
\end{definition}

Let $A = (Q, q_0, F, \Delta, V, n, M_0)$ be an $n$-\lama and $w= u_1\dots u_m \in \alphab^*$. 
We know that $w  \in \lang{A}$ if and only if there is a finite  sequence of transitions  such that for some  $M_f$, $(q_0,M_0) \xRightarrow[A]{w} (q_f,M_f)$ with $q_f \in F$.
It is then possible to construct a \nuauto that accepts $\xi^n(w)$, which simulates the recognition process of the $n$-\lama over the word $w$. 
To do so, we encode every observable transition of $A$ into a sequence of transitions successively simulating the constraints applied to variables  of each layer. Moreover,  we apply the  renaming function $\xi^n$ to the initial memory context.
In order to simplify the notations, in the following, we denote by
$\proj{x}{k}$ the projection onto the $k$-th element of tuple $x$, e.g., $\proj{(a,b,c)}{2} = b$.

\begin{definition}[Encoding of a memory context]
Let $M$ be the memory context over the set of variables $V$ over $n$ layers, 
then $\forall v^l \in V$, its renaming through $\xi^n$, is defined as 
$\enco{M}_\xi(v^l) = \{ \proj{\xi^n(u)}{l} \mid u \in M(v^l)\}$. 
\end{definition}

\begin{definition}[Encoding of a \lama]
Let $A = (Q, q_0, F, \Delta, V, n, M_0)$ be an $n$-\lama, then the \nuauto $\enco{A}_\xi = (Q', q_0', F', \Delta', V', M_0')$
is the encoding of $A$ through the renaming $\xi^n$,
where:
\begin{itemize}
    \item $Q' = Q \cup Q_o$ and the set of states $Q_o = \{ q_\delta^l \mid \delta = (q,\alpha, q') \in \Delta, l \in [2,n] \}$ is used by the sequence of transitions simulating each observable transition of $A$, 
  ${q_0}' = q_0$ and $F' = F$;
    \item $V' = V $ is the set of variables of $A$ flattened on one layer;
    \item ${M_0}'= \enco{M_0}_\xi$ is the initial memory context of $A$ renamed in case it is not initially empty; 
    \item $\Delta' = {\Delta'}_o \cup \setsilent$ where $\setsilent$
    is the set of  all  non-observable transitions of $A$, and 
    ${\Delta'}_o$ contains the encoding of every observable transition $\delta = (q, \alpha, q')$ of $A$, which is a  sequence of transitions 
    with $q_\delta^1=q$ and $q_\delta^n=q'$ such that
   
    $$\begin{array}{lll}{\Delta'}_o & = & 
    \{(q_\delta^l,v^l,\mathbf{x},q_\delta^{l+1}) \mid \delta = (q,  \alpha, q') \in \Delta, l\in [1,n{-}1], \alpha(l)=(v^l, \mathbf{x})\}
    \\&\cup & \{(q_\delta^l,\sharp, q_\delta^{l+1}) \mid \delta = (q,  \alpha, q') \in \Delta, l\in [1,n{-}1], \alpha(l)=\sharp  \}.
    \end{array}$$
\end{itemize}

\end{definition}
Notice that the language accepted by the encoded \nuauto  $\enco{A}_\xi$   of a \lama $A$ is an over-approximation of the language accepted by $A$:   $\xi(\lang{A}) \subseteq \lang{\enco{A}_\xi}$. 
They are  equal only when the \lama has  one layer (i.e., $n=1$).
Nonetheless, this construction may be used to test the membership of a word $w$ to $\lang{A}$. The proof is a simple induction on the length of the derivation of $w$ and $\xi^n(w)$.

\begin{restatable}{theorem}{lamatonu}
    Let $A$ be an $n$-\lama.
$w \in \lang{A}$ if and only if $\xi^n(w) \in \lang{\enco{A}_\xi}$.
\end{restatable}

\subparagraph*{Simulating the membership for \hra in \lama}

This section presents how to solve the membership problem for \hra using \lama.
The difference in expressiveness between \hra and \lama comes from the ability of \hra of removing letters from histories  when they are read.
We resort to an encoding of words where each letter is duplicated and annotated with a number representing how many occurrences of that letter have been encountered so far. 
In detail, the first copy of the letter keeps the information on the number of occurrences of the letter seen so far and the second one the number of occurrences including  the present one. 
\begin{example}
    Take $w =    a   b   a   c  a $ then the encoded-word  is
    $w' = a^0 a^1 ~ b^0 b^1 ~ a^1 a^2 ~c^0 c^1 ~a^2 a^3$
\end{example}

The idea behind the encoding of observable transitions is to use the first copy to check the presence and absence 
of the letter in every variable (simulating the role of  $H_r$) while the second one (that is always fresh) can be used to simulate writing and removal (hence simulating $H_w$).
More precisely, once we add an annotated letter  to a variable, the encoded automaton will ensure that the variable always stores the last seen occurrence of that letter. Thus, removing a letter from a history consists  in not  storing the last seen occurrence of the letter in the corresponding encoded variable.
Clearly,  all the  letters annotated with a number smaller than the current one will not be used in any of the transitions,   representing a form of garbage.

We consider a renaming function $\zeta_i : \alphab \rightarrow \alphab^2$ which replaces $u$ by a pair of letters  $u^{i-1} u^{i}$ for any $i\in \mathbb{N^+}$. Then, we define  the encoding of words $\zeta : \alphab^* \rightarrow \alphab^{*}$ as follows
$\zeta(u_1\ldots u_m)=\zeta_{i_{u_1}}(u_1)\ldots\zeta_{i_{u_m}}(u_m)$ 
where each $i_{u_j}$ is the  number of occurrences of $u_j$ seen
so far. Notice that when considering the word up to letter $u_i$, $\proj{\zeta_{i_{u_i}}(u_i)}{2}$ is always a new letter (e.g., a fresh letter with respect to those in $\zeta(u_1\ldots u_i)$).

\begin{definition}[Encoding of an \hra]\label{def:enco_hra_lama}
Let $A = (Q, q_0, F, \setobs \cup \setsilent, \{h_1,\dots h_n\}, M_0)$ be an \hra, 
its encoding into an $n$-\lama is $\enco{A}_\zeta = (Q', q_0', F', \Delta', V', n, M_0')$ where:
\begin{itemize}
 \item $Q'=Q\cup Q_o$    and the set of states $Q_o = \{q_\delta \mid q\in Q, \delta = (q, H_r, H_w,q') \in \Delta\}$ is used by the sequence of transitions simulating each observable transition of $A$;
\item $q_0'=q_0$ and $F'=F$;
 \item $V'= \{h^l,\omega^l \mid l \in [1,n]\}$ and for each layer $l\in [1,n]$, $h^l$ plays the role of history $h_l$ and  $\omega^l$ is used to check the absence of  letters in $h^l$.
 \item $M_0'(h^l)= \{ \proj{\zeta_1(u)}{1} \mid u \in M_0(h_l) \}$ and  $M_0'(\omega^l) = \emptyset$ for all $l\in [1,n]$ meaning that $M_0'$ is as  $M_0$ with all letters renamed with $\zeta_1$ and empty for all extra variables;  
 \item $\Delta' = \setsilent' \cup \setobs' $  with
    \begin{itemize}
        \item $\setsilent' = \{(q, \{h^l\mid h_l \in H_\emptyset\},q') \mid (q,H_\emptyset,q') \in \setsilent\}$, is the direct translation of the $\varepsilon$-transitions in $A$.
    
       \item $\setobs' = \{ (q, \alpha_{H_r}, q_\delta), (q_\delta, \alpha_{H_w} ,q') \mid \delta = (q,H_r,H_w,q') \in \setobs \}$ with for all $l\in [1,n]$
       \[  
       \begin{array}{lcr}
       \alpha_{H_r}(l)=\left\{ \begin{array}{ll}
            (h^l,\reads)  & \text{if } h_l \in H_r \\
            (\omega^l,\writes) & \text{if } h_l \not\in H_r
             \end{array} \right.
    &
    \mbox{ and }
    &
    \alpha_{H_w}(l) =\left\{ \begin{array}{ll} 
            (h^l,\writes)  & \text{if } h_l \in H_w \\
            \sharp & \text{if } h_l \not\in H_w
        \end{array} \right.
       \end{array}
       \]
       the first simulating the guard part of the observable transition and the second the writing/relocation.
\drop{
$$\alpha_{H_r}(l)=\left\{ \begin{array}{ll}
            (h^l,\reads)  & \text{if } h_l \in H_r \\
            (\omega^l,\writes) & \text{if } h_l \not\in H_r
             \end{array} \right.$$ 
simulating the guard part of the observable transition and
$$\alpha_{H_w}(l) =\left\{ \begin{array}{ll} 
            (h^l,\writes)  & \text{if } h_l \in H_w \\
            \sharp & \text{if } h_l \not\in H_w
        \end{array} \right. $$ 
simulating the writing/relocation.  
}
    \end{itemize}
\end{itemize}

\end{definition}

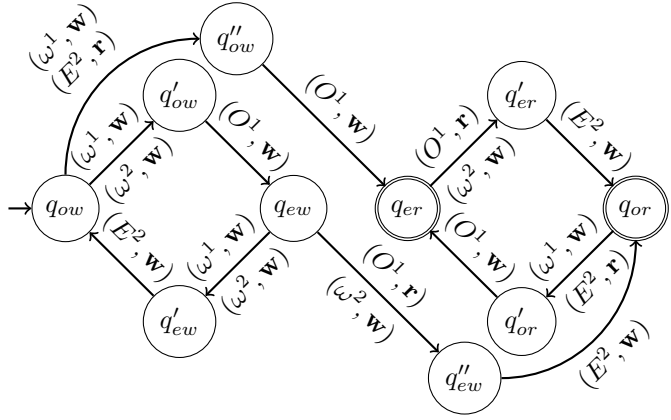
\begin{figure} 
    \centering
    \begin{tikzpicture}[scale=1.5]
\node [draw, circle] (qow1) at (0,0) {$q_{ow}$};
\draw [->, thick](-0.5,0) --(qow1);
\node [draw, circle] (qow2) at (1,1) {$q'_{ow}$};
\node [draw, circle] (qow3) at (1.5,1.5) {$q''_{ow}$};
\node [draw, circle] (qew1) at (2,0) {$q_{ew}$};
\node [draw, circle] (qew2) at (1,-1) {$q'_{ew}$};
\node [draw, circle] (qew3) at (3.5,-1.5) {$q''_{ew}$};
\node [draw, circle, double] (qer1) at (3,0) {$q_{er}$};
\node [draw, circle] (qer2) at (4,1) {$q'_{er}$};
\node [draw, circle, double] (qor1) at (5,0) {$q_{or}$};
\node [draw, circle] (qor2) at (4,-1) {$q'_{or}$};

\draw [->, thick] (qow1) to node[above, rotate = 45] {$(\omega^1, \writes)$} node[below, rotate = 45] {$(\omega^2, \writes)$} (qow2);
\draw [->, thick] (qow2) to node[above, rotate = -45] {$(O^1, \writes)$} (qew1);
\draw [->, thick] (qew1) to node[above, rotate = 45] {$(\omega^1, \writes)$} node[below, rotate = 45] {$(\omega^2, \writes)$} (qew2);
\draw [->, thick] (qew2) to node[above, rotate = -45] {$(E^2, \writes)$} (qow1);
\draw [->, thick] (qow1) to[out = 90, in = 180] node[above, rotate = 45] {$\begin{array}{l}(\omega^1, \writes)\\(E^2, \reads) \end{array}$} (qow3);
\draw [->, thick] (qow3) to node[above, rotate = -45] {$(O^1, \writes)$} (qer1);
\draw [->, thick] (qew1) to node[above, rotate = -45] {$(O^1, \reads)$} node[below, rotate = -45] {$(\omega^2, \writes)$} (qew3);
\draw [->, thick] (qew3) to [out = 0, in = -90] node[below, rotate = 45] {$(E^2, \writes)$} (qor1);
\draw [->, thick] (qer1) to node[above, rotate = 45] {$(O^1, \reads)$} node[below, rotate = 45] {$(\omega^2, \writes)$} (qer2);
\draw [->, thick] (qer2) to node[above, rotate = -45] {$(E^2, \writes)$} (qor1);
\draw [->, thick] (qor1) to node[above, rotate = 45] {$(\omega^1, \writes)$} node[below, rotate = 45] {$(E^2, \reads)$} (qor2);
\draw [->, thick] (qor2) to node[above, rotate = -45] {$(O^1, \writes)$} (qer1);

\end{tikzpicture}
    \caption{The 2-\lama $\enco{A_h}_\zeta$, encoding of the HRA $A_h$ from Fig. ~\ref{fig:HRA_example}. }
    \label{fig:Member_HRA_Lama}
\end{figure}

\begin{example} 
Fig. ~\ref{fig:Member_HRA_Lama} depicts the encoding applied to the \hra of Example~\ref{ex:hra}. Given the word $w = abcabba$ and its renaming $\zeta(w) = a^0 a^1 b^0 b^1 c^0 c^1 a^1 a^2 b^1 b^2 b^2 b^3 a^2 a^3$, we present how the transitions of $\enco{A}_\zeta$ encode the ones in $A$.
$A$ has two histories:  $H = \{O,E\}$, thus the set of variables $V$ of  $\enco{A}_\zeta$ is  $\{O^1, \omega^1\}$ on layer $1$ and $\{E^2, \omega^2\}$ on layer $2$.
Let $(q_{ow}, M_\emptyset)$ be the initial state for both automata, with $M_\emptyset$ the memory context where all variables/histories are empty.

When reading the first letter $a$ in $A_h$, only transition $(q_{ow}, M_\emptyset) \xrightarrow[A_h]{a} (q_{ew}, M_1)$ can be enabled as $a$ is not stored in the histories in $M_\emptyset$, as a consequence, $a$ is added to $O$ in $M_1$. The transition from $q_{ow}$ to $q_{er}$ cannot be enable as $a$ is not store in $E$.
In $\enco{A}_\zeta$, this transition is encoded with the  sequence $(q_{ow},M_\emptyset) \xrightarrow[\enco{A_h}_\zeta]{a^0} (q_{ow}',M') \xrightarrow[\enco{A_h}_\zeta]{a^1} (q_{ew},M'_1)$. The first transition when reading $a^0$, checks if $a^0$ is absent from both $O^1$ and $E^2$ using  $\omega^1$ and $\omega^2$ with the injectivity constraint. When reading $a^1$ the transition $q_{ow}'\rightarrow q_{ew}$ writes the letter in $O^1$. Thus, in $M_1'$,  $a^1$ belongs to $O^1$, as $a^1$ is the last occurrence of $a$. Note that $a^0$ is still stored in $\omega^1$ and $\omega^2$, but it will never be read again (as the renaming $\zeta$ always increases the index of letters).

Similarly, when $A_h$ read the first occurrence of $b$, the only transition enable is $(q_{ew},M_1) \xrightarrow[A_h]{b} (q_{ow}, M_2)$, where $b$ is stored in $E$ is $M_2$. And when reading $c$ the only transition enabled is $(q_{ow},M_2) \xrightarrow[A_h]{c} (q_{ew}, M_3)$ with $O$ storing both $a$ and $c$ while $E$ only stores $b$.
In $\enco{A}_\zeta$, this sequence of transitions is encoded by enabling the sequence of transitions $(q_{ew}, M_1')\xrightarrow[\enco{A_h}_\zeta]{b^0}(q_{ew}', M_1'')\xrightarrow[\enco{A_h}_\zeta]{b^1}(q_{ow}, M'_2)\xrightarrow[\enco{A_h}_\zeta]{c^0}(q_{ow}', M''_2)\xrightarrow[\enco{A_h}_\zeta]{c^1}(q_{ew}, M'_3)$. With $M'_2$ storing $b^1$ in $E^2$ and $M'_3$ storing $c^1$ in $O^3$ in addition to $a^1$. This is the only sequence of transition that can be enabled as $b_0$ was not stored in $O^1$ in the state $(q_{ew}, M'_1)$ and $c^0$ was not stored in $E^2$ in $(q_{ow}, M'_2)$.

When reading the second occurrence of $a$ in $A_h$, the only enabled transition is $(q_{ew}, M_3) \xrightarrow[A_h]{a} (q_{or}, M_4)$ where $q$ was transferred from $O$ to $E$ in $M_4$. In $\enco{A_h}_\zeta$ this is encoded by sequence of transition $(q_{ew}, M'_3)\xrightarrow[\enco{A_h}_\zeta]{a^1} (q_{ew}'',M_3'')  \xrightarrow[\enco{A_h}_\zeta]{a^2}(q_{or},M_4')$. The first transition of this sequence is the only enabled one in configuration $(q_{ew}, M'_3)$ as $a^1$ is already stored in $O^1$, thus it would be impossible to write it in $\omega^1$ to enable the transition to $q_{ew}'$. In $M_4'$, the letter $a^2$ is stored in $E^2$ along with $b^1$, while $a^1$ is still stored in $O^1$ but will never be read again in $\zeta(w)$, so it can be ignored. This is how the transfer mechanism is encoded in this construction.

Reading $bb$, the last two letters of $w$, will enable in $A$ the sequence $(q_{or}, M_4) \xrightarrow[A]{b}(q_{er},M_5)$ transferring $b$ from $E$ to $O$ in $M_5$ and then enabling $(q_{er}, M_5) \xrightarrow[A]{b}(q_{or},M_6)$ transferring $b$ back from $O$ to $E$ in $M_6$.
In $\enco{A_h}_\zeta$, this is encoded by reading the letters $b^1b^2b^2b^3$ and enabling the loop of transition between states $q_{or}$, $q'_{or}$, $q_{er}$ and $q'_{er}$. Looking if the previous occurrence of $b$, here $b^2$ (resp. $b^3$), is stored in $E^2$ (resp. $O^1$) by reading in the variable. Also checking if it is absent from $O^1$ (resp. $E^2$) by writing in the $\omega$ of the same layer. Then writing the next occurrence of $b$, here $b^3$ (resp. $b^4$), in $O^1$ (resp. $E^2$) to encode its transfer.  

\end{example}

The proof that the encoding is correct can be found in Appendix \ref{app:simul-member-hra-lama}. Notice that, as before, the language recognised by $\enco{A}_\zeta$ is actually larger than $\lang{A}$.

\begin{remark} 
In \cite{GT2016} the HRA are presented with a set of registers, able to store only one letter at a time. Its content is overwritten whenever a letter is written into it. The author proved that HRA using only histories are as expressive as the ones using both histories and register. However, the construction presented to remove register is exponential in their number. 
This is caused by the necessity to erase the content of histories, to encode the overwriting, before reading their content to verify if an observable transition is enabled.
The construction of \ref{def:enco_hra_lama} split the enabling check and writing of the letter into two observable transition. It can be extended to register by adding an intermediate non-observable transition erasing the content of variables modeling registers to simulate overwriting. 
\end{remark}

\begin{restatable}{theorem}{hratonu}
    Let $A$ be an \hra.
$w \in \lang{A}$ if and only if $\zeta(w) \in \lang{\enco{A}_\zeta}$.
\end{restatable}

\subparagraph*{Complexity}
The two previous encodings give polynomial reductions of the membership problem from \hra to \lama and from \lama to \nuauto. Therefore, there is a polynomial reduction of the problem for \hra to \nuauto.
The expressiveness results from \cite{BertrandKP22} give a linear construction from \nuauto to \lama and an exponential construction, in the number of layers, from \lama to \hra.
As \nuauto are 1-\lama, the same construction can be used to translate a \nuauto into an HRA of polynomial size.
This implies an equivalence of complexity class of the membership problem for \nuauto and \hra, as well as for \nuauto and \lama. By transitivity we get the same equivalence between \lama and \hra.
Next, we show that the membership problem for \lama is NP-complete. For the hardness part, this is shown by resorting to a reduction from the 3SAT problem, while the completeness part follows by observing what would be the cost of executing a word on an automaton. Fig.  \ref{fig:gadgetsmember} depicts the intuition behind the encoding of a 3SAT instance. The idea is that the gadget on Fig.  \ref{fig:init-xi} chooses non-deterministically the truth assignment of $X_i$ or $\overline{X_i}$ and the one in Fig.  \ref{fig:clause-3SAT} checks that this assignment indeed satisfies the given clauses. The full proof is in Appendix \ref{app:complexity-member-lama}.

\begin{figure}
    \centering
\begin{subfigure}{0.35\textwidth}
\centering
\begin{tikzpicture}[>=stealth',shorten >=1pt,auto,node distance=2.5cm, xscale=2, yscale=.9]

\node[state] (q0)      {$q_0$};
\draw [->](-0.5,0)--(q0);
\node[state, accepting, right of=q0] (qf)      {$q_f$};

\path[->] (q0)  edge [bend right,below]  node {$(X_i, \writes)$} (qf);
\path[->] (q0)  edge [bend left,above]  node {$(\overline{X_i},\writes)$} (qf);

\end{tikzpicture}
    \caption{The gadget for the existentially quantified variable $x_i$}
    \label{fig:init-xi}
\end{subfigure} \hspace*{0.55cm}
\begin{subfigure}{0.5\textwidth}
\centering
\begin{tikzpicture}[xscale=1.3, yscale=.8]
\node [draw, circle] (q0) at (6,0) {$q_0$};
\draw [->](5.5,0)--(q0);
\node [draw, circle, double] (qf) at (10,0) {$q_f$};
\draw [->, thick, rounded corners] (q0)--(6.5,1)--node[above]{$(L_j^1, \reads)$}(9.5,1)--(qf); 
\draw [->, thick] (q0)--node[above]{$(L_j^2,\reads)$}(qf);
\draw [->, thick, rounded corners] (q0)--(6.5,-1)--node[above]{$(L_j^3, \reads)$}(9.5,-1)--(qf);
\end{tikzpicture} 
\caption{The 3SAT clause gadget for $C_j=(L_j^1\vee L_j^2\vee L_j^3)$}
    \label{fig:clause-3SAT}
\end{subfigure} 
    \caption{Gadgets used for showing NP-hardness of the membership problem.}
    \label{fig:gadgetsmember}
\end{figure}
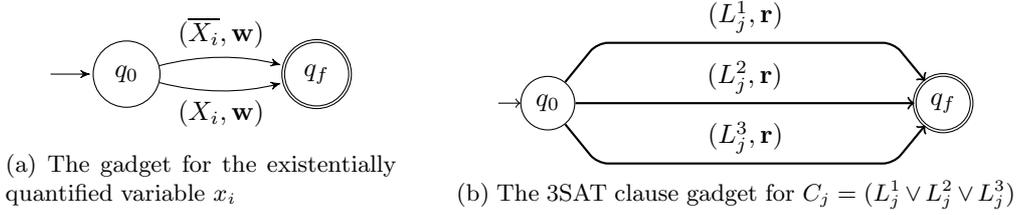

\begin{restatable}{theorem}{lamanp} \label{the:membership-npcomplete}
The membership problem for \lama  is NP-complete. 
\end{restatable}

Hence we can conclude that:
\begin{corollary}
    The membership problems for \nuauto, \lama and \hra is NP-complete.
\end{corollary}

As a direct consequence and looking at the expressiveness hierarchy in Fig.  \ref{fig:graphe_expre} we can also give a complexity class for the membership problem in FRA. Indeed, since FMA can be encoded into FRA \cite{Tzevelekos11}, we can deduce NP-hardness, and completeness follows from their encoding into \lama \cite{BertrandKP22}.

\begin{corollary}
The membership problem for FRA  is NP-complete.
\end{corollary}

\section{Complexity of the non-emptiness problem}\label{sec:empty}

The non-emptiness problem consists in  deciding whether the language accepted by  an automaton is non-empty, or in other words checking if there is a path from the initial configuration to a final configuration.
As mentioned before, in \cite{GT2016}, it has been shown that deciding the non-emptiness for \hra is Ackermann-complete. Still, the complexity for non-emptiness is known neither for \lama nor for \nuauto. 
We start with the Non-Emptiness Problem for \nuauto. We show that the problem is PSPACE-complete. To do so, we reduce the TQBF problem  (true fully quantified Boolean formula) to \nuauto non-emptiness. TQBF is known to be  PSPACE-complete (Meyer-Stockmeyer theorem \cite{arora2006computational}).

\begin{restatable}{lemma}{nunep}\label{lem:empty-hard}
    The Non-Emptiness problem is  PSPACE-hard for \nuauto. 
\end{restatable}
\begin{proof}
Let $\nu$NEP be the short for Non-Emptiness Problem for \nuauto. 

We show that TQBF can be reduced to  $\nu$NEP.
Let $Q_1 x_1 \ldots Q_n x_n (C_1 \wedge \ldots \wedge C_m)$, be a fully quantified Boolean formula, where each $Q_i\in \{\forall, \exists \}$ and each $C_j$ is a clause comprising at most $n$ 
literals ($x_i$ or $\overline{x_i}$). 
We assume that literals in clauses are ordered according to the order of variable declarations and at most one literal per variable is present.
To encode TQBF in \nuauto we consider: 
\begin{itemize}
    \item for each existentially quantified $x_i$, variables $X_i$ and $\overline{X_i}$, and the gadget depicted in Fig. ~\ref{fig:init-xi}, used in the proof of Theorem ~\ref{the:membership-npcomplete};
    \item for each universally quantified $x_i$, variables $X_i$, $\overline{X_i}$ and $\widetilde{X_i}$, and the gadget depicted in Fig. ~\ref{fig:univ}. 
    Variable $\widetilde{X_i}$ is used as a flag to indicate that all possible truth assignments of $x_i$ have been considered.
    The initial transition of the gadget initialises variable $\overline{X_i}$ to $1_i$.  
    We do not detail the precise automaton handling other variables and the clauses as it will be constructed recursively (and explained below), connected between states $q_1$ and $q_2$. 
    The looping part starting in state $q_2$ writes letter $2_i$ into variable $\widetilde{X_i}$, which is used after browsing once again the dashed part (corresponding to variables and clauses) to reach the final state $q_f$. 
    After this, variable $\overline{X_i}$ is reset and then variable $X_i$ is initialized to $1_i$ in order to consider the other truth assignment of $x_i$. From state $q_5$ to $q_1$ all the variables for $x_j$, with $j$ from $i+1$ to $n$ are reset to reinitialize their truth assignments; 
    \item for each clause $C_j$, a clause gadget depicted in Fig. ~\ref{fig:clause-TQBF}. It tests  literals one after the other and takes the oblique transition for the first which makes the clause satisfied, which means that  the remaining literals are just read up to the end of the clause, which is satisfied if $q_f$ is reached.
\end{itemize} 

\begin{figure}
\begin{subfigure}{\textwidth}
\centering
\begin{tikzpicture}[>=stealth',shorten >=1pt,auto,node distance=2cm, xscale=2.2, yscale=1]

\node[state] (q0)      {$q_0$};
\node [left of = q0, node distance = 1cm] (start) {};
\draw [->] (start)--(q0);
\node[state] (q1) [right of=q0, very thick]  {$q_1$};
\node[rectangle, rounded corners, draw, dashed, text width=1.8cm, minimum height=1cm, minimum width=1.5cm, align=center,node distance=2cm] (var) [right of=q1] {variables};
\node[rectangle, rounded corners, draw, dashed, text width=1.8cm, minimum height=1cm, minimum width=1.5cm, align=center,node distance=2.5cm] (cla) [right of=var] {clauses};
\node[state] (q2) [right of=cla, very thick]  {$q_2$};
\node[state] (q3) [above of=q2, node distance = 1.7cm]  {$q_3$};
\node[state] (q4) [left of=q3, node distance=3cm]  {$q_4$};
\node[state] (q5) [above of=q1, node distance = 1.7cm]  {$q_5$};
\node[state, accepting] (qf) [right of=q2]  {$q_f$};

\path[->] (q0)  edge  node {$(\overline{X_i},\writes)$} (q1);
\path[->] (q2)[right]  edge  node {$(\widetilde{X_i},\writes)$} (q3);

\path[->] (q4)[above]   edge  node {$(X_i,\writes)$} (q5);
\path[->] (q2)  edge  node {$(\widetilde{X_i},\reads)$} (qf);
\path[->] (q3)[above]  edge  node {$\rset(\overline{X_i})$} (q4);
\path[->] (q5)[right]  edge  node {$\rset(\{X_j,\overline{X_j},\widetilde{X_j}\mid j>i \})$} (q1);

\path[->, dashed] (q1)  edge  (var);
\path[->, dashed] (var)  edge  (cla);
\path[->, dashed] (cla)  edge  (q2);

\end{tikzpicture} 
    \caption{The gadget for universally quantified variable $x_i$}
    \label{fig:univ}
\end{subfigure}
\vspace{0.2mm}

\begin{subfigure}{\textwidth}
\centering
\begin{tikzpicture}[>=stealth',shorten >=1pt,auto,node distance=2.5cm, xscale=2, yscale=.5]

\node[state] (q0)      {$q_0$};
\node [left of = q0, node distance = 1cm] (start) {};
\draw [->] (start)--(q0);
\node[state] (q1p) [right of=q0] {$q_1'$};
\node[state] (q1) [above of=q1p, node distance = 1.6cm] {$q_1$};
\node[state] (q2) [right of=q1]  {$q_2$};
\node[state] (q3) [right of=q2] {$q_3$};
\node (d1) [right of=q3, node distance = 1cm] {$\cdots$};
\node[state,accepting] (q4) [right of=d1, node distance = 1cm] {$q_k$};

\node[state] (q2p) [right of=q1p] {$q_2'$};
\node[state] (q3p) [right of=q2p] {$q_3'$};
\node (d2) [right of=q3p, node distance = 1cm] {$\cdots$};
\node[state] (q4p) [right of=d2, node distance = 1cm] {$q_k'$};

\path[->] (q0)  edge [right, near start] node {$(L_{j_1},\reads)$} (q1);
\path[->] (q1)  edge [bend right,below]  node {$(L_{j_2},\reads)$} (q2);
\path[->] (q1)  edge [bend left,above]  node {$\overline{(L_{j_2}},\reads)$} (q2);
\path[->] (q2)  edge [bend right,below]  node {$(L_{j_3},\reads)$} (q3);
\path[->] (q2)  edge [bend left,above]  node {$(\overline{L_{j_3}},\reads)$} (q3);

\path[->] (q0)  edge  [below] node {$(\overline{L_{j_1}},\reads)$} (q1p);
\path[->] (q1p)  edge [right, near start]  node {$(L_{j_2},\reads)$} (q2);
\path[->] (q1p)  edge [below]  node {$(\overline{L_{j_2}},\reads)$} (q2p);
\path[->] (q2p)  edge [right, near start]  node {$(L_{j_3},\reads)$} (q3);
\path[->] (q2p)  edge [below]  node {$(\overline{L_{j_3}},\reads)$} (q3p);

\end{tikzpicture}
    \caption{The gadget for TQBF clause $C_j=(L_{j_1}\vee \ldots \vee L_{j_k})$.}
    \label{fig:clause-TQBF}
\end{subfigure}

    \caption{Gadgets used for showing NP-hardness of the emptiness problem for \nuauto.}
    \label{fig:gadgets}
\end{figure}
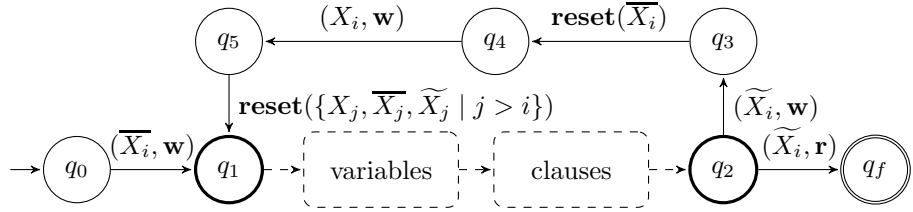
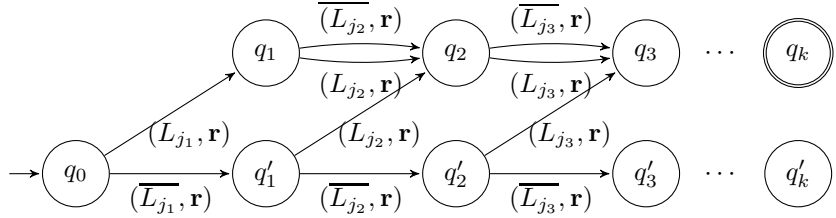

In order to construct the \nuauto $A$ encoding the instance of TQBF we connect 
first the clause gadgets by merging the final state of a clause gadget with the initial state of the next one, let  $C$ be the resulting automaton. Then, we connect to $C$ the gadgets for variable declarations starting from the $n$th, i.e., the last in the order of declarations. If the variable is under an existential quantifier, we connect the existential variable gadget  in front of the automaton obtained so far by merging its final state with the initial state of $C$. If the variable is under a universal quantifier, we connect the corresponding gadget by merging the initial state of $C$ with state $q_1$ of the gadget, and the final state of $C$ with state $q_2$ of the gadget. We connect this way, i.e., following the inverse order of declarations,  the  gadgets for all the remaining variable declarations.
The initial state of the first declared variable gadget is the initial state of $A$ and the unique state $q_f$ of the final construction is the unique final state of $A$. 
Finally, the input word $w$ is obtained recursively for each TQBF instance by the function $input(\phi)$ defined in Algorithm \ref{algo-input}. The construction of the word follows the intuition given above (for the construction of the automaton), that is it unfolds the loops generating the letters needed at each step.

\begin{algorithm}
\begin{algorithmic}[1]

\Function{input}{$\phi$}       \Comment{$\phi= Q_1 x_1 \ldots Q_n x_n \ C_1 \ldots C_m$}
\State $\forall i\in [1,n]: \init(x_i) = 1_i$
\State $\forall i\in [1,n]: \done(x_i) = 2_i$
\State $\forall j\in [1,m]: w_j$ 
    \Comment{contains exactly one $1_i$ for each $x_i$ or $\overline{x_i}$ present in clause $C_j$}

    \If{$\phi=\emptyset$}
        \State \Return  $\epsilon$
    \ElsIf{$\phi=\exists x_i \ \phi'$}
        \State \Return  $\init(x_i).input(\phi')$
        
    \ElsIf{$\phi=\forall x_i \ \phi'$}
        \State \Return  $\init(x_i).input(\phi').\done(x_i).\init(x_i).input(\phi').\done(x_i)$
    \ElsIf{$\phi=C_i \ \phi'$}
        \State \Return  $w_i. input(\phi')$
        
    \EndIf
\EndFunction
\end{algorithmic}
\caption{Function to generate the word accepted by TQBF automaton}
\label{algo-input}
\end{algorithm}


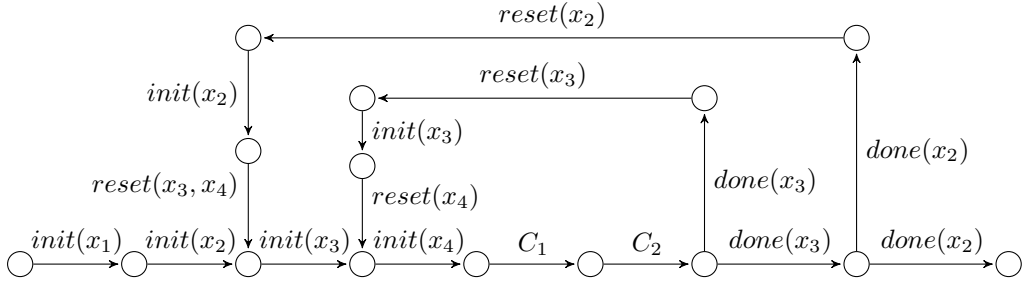
\begin{figure}
    \centering 
\begin{tikzpicture}[>=stealth',shorten >=1pt,auto,node distance=1.5cm, xscale=0.5, yscale=.5]

\node[draw, circle] (q0)      {};
\node[draw, circle] (q1) [right of=q0]{};
\node[draw, circle] (q2) [right of=q1]{};
\node[draw, circle] (q3) [right of=q2]{};
\node[draw, circle] (q4) [right of=q3]{};
\node[draw, circle] (q5) [right of=q4]{};
\node[draw, circle] (q6) [right of=q5]{};
\node[draw, circle] (q7) [right of=q6,node distance=2cm]{};
\node[draw, circle] (q8) [right of=q7,node distance=2cm]{};

\path[->] (q0)  edge  node {$\init(x_1)$} (q1);
\path[->] (q1)  edge  node {$\init(x_2)$} (q2);
\path[->] (q2)  edge  node {$\init(x_3)$} (q3);
\path[->] (q3)  edge  node {$\init(x_4)$} (q4);
\path[->] (q4)  edge  node {$C_1$} (q5);
\path[->] (q5)  edge  node {$C_2$} (q6);
\path[->] (q6)  edge  node {$\done(x_3)$} (q7);
\path[->] (q7)  edge  node {$\done(x_2)$} (q8);

\node[draw, circle] (q9) [above of=q6,node distance=2.2cm]{};
\node[draw, circle] (q10) [above of=q3,node distance=2.2cm]{};
\node[draw, circle] (q10a) [above of=q3,node distance=1.3cm]{};
\path[->] (q6)[right]  edge  node {$\done(x_3)$} (q9); 
\path[->] (q9)[above]  edge  node {$\res(x_3)$} (q10); 
\path[->] (q10)  edge  node {$\init(x_3)$} (q10a); 
\path[->] (q10a)[near start]  edge  node {$\res(x_4)$} (q3);

\node[draw, circle] (q11) [above of=q7, node distance=3cm]{};
\node[draw, circle] (q12) [above of=q2, node distance=3cm]{};
\node[draw, circle] (q13) [below of=q12]{};
\path[->] (q7)[right]  edge  node {$\done(x_2)$} (q11); 
\path[->] (q11)[above]  edge  node {$\res(x_2)$} (q12); 
\path[->] (q12)[left]  edge  node {$\init(x_2)$} (q13); 
\path[->] (q13)[left, near start]  edge  node {$\res(x_3,x_4)$} (q2);

\end{tikzpicture}

\caption{Schema of construction for TQBF instance $\exists x_1 \ \forall x_2 \ \forall x_3 \ \exists x_4 \ ((x_1 \vee \overline{x_4}) \wedge (\overline{x_2} \vee \overline{x_3} \vee x_4))$.}
\label{fig:schema}
\end{figure}

Example: for the TQBF instance $\exists x_1 \ \forall x_2 \ \forall x_3 \  \exists x_4 ((x_1 \vee \overline{x_4}) \wedge (\overline{x_2} \vee \overline{x_3} \vee x_4))$, Fig. ~\ref{fig:schema} represents a general construction schema of the corresponding \nuauto, and the input word is 
\drop{$w=\init(x_1).\init(x_2).\init(x_3).\init(x_4).w_1.w_2.$ 
$\init(x_3).\done(x_3).\init(x_4).w_1.w_2. \ \done(x_3).$ \\
$\init(x_2).\done(x_2).\init(x_3).\init(x_4).w_1.w_2.$ 
$\init(x_3).\done(x_3).\init(x_4).w_1.w_2. \ \done(x_3).$ \\ $\done(x_2)$
which finally gives: \\
}
$1_1 1_2  1_3 1_4     1_1 1_4  1_2 1_3 1_4     2_3 
1_3  1_4              1_1 1_4  1_2 1_3 1_4     2_3 
2_2 
1_2  1_3  1_4         1_1 1_4  1_2 1_3 1_4     2_3  
1_2 1_3  1_4          1_1 1_4 1_2 1_3 1_4      2_3  2_2$. 

Note that every gadget of the automaton is deterministic, except for the existential variable gadget. 
The size of $A$ is polynomial in the size of the TQBF expression. The length of the word generated by Algorithm \ref{algo-input} is in $\Omega(2^n)$ 
but it is not a parameter of the construction of $A$.
Clearly, only a word generated by Algorithm \ref{algo-input} (or a consistent renaming) can be accepted by $A$ starting with an empty memory context. Such a word can be accepted if and only if there is a solution to the TQBF instance.
\end{proof}

It remains to show that the non-emptiness problem for \nuauto is in PSPACE.
This accounts for showing that if the language recognized by an \nuauto is non-empty then it contains a word whose size together with  the length of the transition path needed to accept it, are exponentially bounded with respect to the size of the \nuauto. To this aim, we  build a finite state machine (FSM) characterizing an abstraction of the  state space of the \nuauto.

The idea is that when one can choose the letter read, observable 
transitions that write a letter in a variable are never blocking. Since the alphabet is infinite there is always a fresh letter that can be added, which we call a token. Instead, observable transitions that read a letter from a variable are blocking, in the sense that concerned variables must contain at least a letter (that we call a key). 
The first step towards the construction of the FSM, is to build a \emph{canonical} \nuauto such that 
a  word accepted by the canonical automaton will also be accepted by the initial \nuauto $A$ and each word accepted by $A$ will have a corresponding canonical version. 
Consider a \nuauto $A=(Q, q_0, F, \Delta, V, M_0)$, its \emph{canonical} version $\simpli(A)=(Q, q_0, F, \Delta, V, M_0')$ is an \nuauto over the  alphabets $K$ and $T$, where:
\begin{itemize}
    \item $K \subset \alphab$, such that $|K|=|V|$, is the set of \emph{keys} $k_v$ each of them being associated with a variable $v \in V$. If $M_0(v)\neq\emptyset$, we select $k_v$ in $M_0(v)$. Also, if $M_0(v)=\emptyset$, we select $k_v$ such that $\forall v'\in V,k_v \not\in M_0(v')$. The presence of a key in a variable  $v$ denotes the fact that $v$ is non-empty.
    \item $T=\{t_1,t_2,\ldots \} \subset \alphab$, $K\cap T=\emptyset$, is an infinite set containing letters
    called \emph{tokens} intended to be used only once,
    thus they are never stored in memory. As a consequence, no letter in $T$ is present in the initial memory context of $A$, $\forall v\in V: M_0(v) \cap T = \emptyset$.  
   \item For each $v\in V$, the initial memory context $M_0'(v)$ of $\simpli(A)$ is either empty if $M_0(v)=\emptyset$, or if $M_0(v)\neq\emptyset$, it only contains its key $k_v$.
\end{itemize}

Notice that a word $w$ is accepted by $\simpli(A)$ if the following conditions hold:
\begin{enumerate}
    \item $w\in (K\cup T)^*$, and if $t_i \in T$ appears in $w$ then it occurs  at most once, 
     \item let  $(q_0,M_0') \xRightarrow[\simpli(A)]{w}(q_f,M_f)$ with $q_f\in F$ be the accepting path for $w$ then for each intermediate configuration $(q,M)$ in the path and for each $k_v\in K$ either $k_v \in M(v)$ or for all $v' \in V$, $k_v \notin M(v')$.  
\end{enumerate} 

Observe that $\simpli(A)$ is actually the same automaton as $A$ but over a subset of the alphabet $\alphab$ and where for each $v\in V$,  $M_0'(v) \subseteq M_0(v)$, hence it is easy to conclude that the language of $\simpli(A)$ is included in the one of $A$. The complete proof is in  Appendix~\ref{app:empty-pspace}.

\begin{restatable}{lemma}{simplisubset}
 \label{lem:simplifiedincludedinnormal}
Let $A$ and $\simpli(A)$ be a \nuauto and its canonical version.
If a word $w \in \lang{\simpli(A)}$ then $w \in \lang{A}$.   
\end{restatable}

We want to show that the language accepted by a \nuauto $A$ is empty if and only if the language accepted by $\simpli(A)$ is empty. The if part is the most involved and is the content of the following lemma.

\begin{lemma}\label{lem:uniform-to_simpli}
Let $A$ be a \nuauto and $\simpli(A)$ and its canonical version.
If  $w \in \lang{A}$ then there exists $w' \in \lang{\simpli(A)}$.
\end{lemma}

\begin{proof}
    Let $w=u_1\dots u_n \in \lang{A}$. Then there exists an accepting path   
    $(q_0,M_0') \xRightarrow[A]{w}(q_f,M_f)$ with $q_f\in F$ and intermediate configurations $(q_i,M_i)_{i\geq 0}$. Depending on those intermediate configurations we build a new word  $w'= u_1'\dots u_n'$  and the corresponding path in $\simpli(A)$ accepting $w'$. 
    For each configuration $(q_i,M_i)$, the construction maintains an invariant:   $\forall v\in V, M'_i(v)= \{k_v\} \text{ if and only if } M_i(v)\neq \emptyset$.
The proof proceeds by induction: 
\begin{description}
    \item[Base case:] By construction the initial configuration of $\simpli(A)$ satisfies the invariant.
    \item[Inductive step:]  We examine the transition $(q_i, M_i) \xrightarrow[u_i]{\delta_i} (q_{i+1}, M_{i+1})$.
    By inductive hypothesis, we know that there exists a sequence of transitions  $(q_0,M'_0) \xRightarrow[\simpli(A)]{u'_1\dots u'_{i-1}} (q_i,M'_i)$  such that $\forall v \in V, M'_i(v) = \{k_v\} \text{ if and only if } M_i(v) \neq \emptyset$. 
    We prove there is a letter $u'_i$ leading to the configuration $(q_{i+1},M'_{i+1})$ satisfying this property, through $\delta_i$ (by construction $A$ and $\simpli(A)$ are defined on the same set of transitions), we list all possible cases:
    \begin{description}

        \item[- $\delta_i = (q_i, \rset, q_{i+1})$:] then $u_i = u'_i = \varepsilon$ and $\delta_i$ will lead to configuration with $M'_{i+1}(v) = \emptyset$ if $v\in \rset$ or $M'_{i+1}(v) = M'_i(v)$ otherwise. Hence satisfying the invariant.
        \item[- $\delta_i = (q_i, v, \reads, q_{i+1})$:] then $M_i(v) \neq \emptyset$ otherwise the transition could not be enabled, so $u'_i = k_v$ and by inductive hypothesis $M'_i(v)=\{k_v\}$. Since the memory context does not change for both automata,  the invariant is satisfied;
        \item[- $\delta_i = (q_i, v, \writes, q_{i+1})$:] if $M_i(v) = \emptyset$, then $u'_i = k_v$, and $k_v$ will be written in  variable $v$ in $M'_{i+1}$ satisfying the invariant. \\
        If  $M_i(v) \neq \emptyset$, then $u'_i = t_i\in T$ is a token and $M'_{i}(v)=M'_{i+1}(v)$, we recall that tokens are not stored in memory. 
        By inductive hypothesis we know that $M'_i(v) =\{k_v\}$ and as $\delta_i$ is a writing transition, then $M_{i+1}\neq \emptyset$, satisfying the invariant.
    \end{description}   
\end{description}
From the previous construction, it  follows immediately  that $w' \in \lang{\simpli(A)}$.    
\end{proof}

Observe that, when reading a word $w \in \lang{\simpli(A)}$,  we only need to store the letters belonging to $K$. Indeed,  tokens in $T$ may occur only once in $w$. This entails that tokens can only enable a write observable transition, while for read transitions keys are sufficient. 
Hence, in practice, tokens do not need to be added to the memory context.
%
%
%
Hence the number of different configurations in $\simpli(A)$ 
is bounded by $|Q|\cdot 2^{|K|}$ as:  
\begin{itemize}
    \item we have $|Q|$ states that can be  encoded on $log|Q|$ bits, and
    \item there are $2^{|K|}$ possibilities to store the presence or not in the memory of letters in $K$ ($2$ possibilities per letter encoded on $1$ bit since each $k_v$ can only be stored in $v$), so in total we need  $|K|$ bits.
\end{itemize}

This shows that the number of configurations is finite. On top of this, as remarked above, transitions over letters in $T$ do not add constraints on the memory context and they can be ignored. Hence the alphabet is now finite and  we can reduce the non-emptiness of FSM to the non-emptiness problem of \nuauto.

\begin{lemma}\label{lem:empty-pspace}
  The Non-Emptiness problem for \nuauto is in PSPACE
\end{lemma}

\begin{proof}
Given a \nuauto $A = (Q, q_0, F, \Delta, V, M_0)$, its canonical form has at most $|Q|2^{|K|}$ configurations. 
The state space of $\simpli(A)$ could be constructed as an FSM by merging all transitions of $A$ writing a token from $T$ going from state $q$ to $q'$ into a unique $\varepsilon$-transition. This way, the FSM would have $O(|\Delta| 2^{|K|})$ transitions as each configuration $(q,M')$ of $\simpli(A)$ has at most as many outgoing transitions as $q$ in $A$.
A formal definition of the construction is in  Appendix~\ref{app:empty-pspace}.

Moreover, if the underlying FSM is non-empty it implies that there is a sequence of at most $O(|\Delta| 2^{|K|})$ transitions from an initial state of $A$ to one of its accepting state. 
Recall that finding a path between two vertexes/states in a graph $(V,E)$ is a problem called PATH which is NL-complete \cite{arora2006computational}. 
The algorithm in logarithmic space for PATH could be adapted to find if there exists a sequence of transitions from an initial state of $A$ to an accepting state. Since this sequence of transitions is exponential in the size of $A$, thus we prove that the problem is in NPSPACE for \nuauto. 
Since PSPACE=NPSPACE \cite{arora2006computational} we show that the Non-Emptiness problem for $\nu$ is in PSPACE. 

The PATH algorithm adapted to our problem memorizes a state of $A$, the memory of $\simpli(A)$ and a counter on $O(\log(|\Delta|)+ {|K|})$ bits. Each time that the counter is augmented by one, a transition starting in the memorized state will be chosen randomly and applied as follows: if this transition is a reset, then the state is updated and the memory is reset. If this transition is a write, then the state is updated and the corresponding key is added to the memory (if not already present). If the transition is a read then either the key is not in the memory and the algorithm halts and rejects or the state is updated. As soon as an accepting state is reached then the algorithm halts and accepts. If the counter reaches its maximum then the algorithm halts and rejects. 
Note that the FSM is not actually constructed in this algorithm, but one of its paths is explored dynamically. 
%
%
%
\end{proof}


\begin{theorem} \label{theo:nunep-psapceComplete}
 The non-emptiness problem for \nuauto PSPACE-complete. 
\end{theorem}
\begin{proof}
By Lemmata \ref{lem:empty-hard} and \ref{lem:empty-pspace}.
\end{proof}

\section{Conclusions}\label{sec:concl}
We have discussed the complexity of membership and non-emptiness for three formalisms \nuauto, \lama and \hra.
We showed that concerning the membership problem, all three kinds of automata fall in the NP-complete class. Emptiness is more delicate. We proved that the non-emptiness problem for \nuauto is PSPACE-complete. The problem remains open for \lama.

As a matter of fact, concerning the non-emptiness problem for \lama we know the lower bound and the upper bound of the complexity class. As a consequence of Theorem \ref{theo:nunep-psapceComplete} and from the expressiveness results in \cite{BertrandKP22}, the complexity is PSPACE-hard.
However, it is a strict lower bound as we are able to construct a LaMA with a shortest accepted word of size in $O(2^{2^n})$ with $n$ the number of variables.
We believe the problem is actually Ackermann-complete. 
In our previous work\cite{BertrandKP22}, we showed an exponential encoding of \lama into \hra for which the non-emptiness problem is shown to be Ackermann-complete in \cite{GT2016}, which gives us the Ackermann class membership. 


As for future work, apart from closing the open problem, we plan to address the expressiveness of \lama. Indeed the number of layers seems to create a hierarchy of expressiveness and complexity.
\newpage
\bibliography{bib}

\begin{thebibliography}{10}

\bibitem{arora2006computational}
S.~Arora and B.~Barak.
\newblock {\em Computational Complexity: A Modern Approach}.
\newblock Cambridge University Press, 2006.
\newblock URL: \url{https://theory.cs.princeton.edu/complexity/book.pdf}.

\bibitem{bertrand:tel-03172600}
Clement Bertrand.
\newblock {\em {Reconnaissance de motifs dynamiques par automates
  temporis{\'e}s {\`a} m{\'e}moire}}.
\newblock Theses, {Universit{\'e} Paris-Saclay}, December 2020.
\newblock URL: \url{https://theses.hal.science/tel-03172600}.

\bibitem{BertrandKP22}
Cl{\'{e}}ment Bertrand, Hanna Klaudel, and Fr{\'{e}}d{\'{e}}ric Peschanski.
\newblock Layered memory automata: Recognizers for quasi-regular languages with
  unbounded memory.
\newblock In Luca Bernardinello and Laure Petrucci, editors, {\em Application
  and Theory of Petri Nets and Concurrency - 43rd International Conference,
  {PETRI} {NETS} 2022, Bergen, Norway, June 19-24, 2022, Proceedings}, volume
  13288 of {\em Lecture Notes in Computer Science}, pages 43--63. Springer,
  2022.
\newblock \href {https://doi.org/10.1007/978-3-031-06653-5\_3}
  {\path{doi:10.1007/978-3-031-06653-5\_3}}.

\bibitem{BertrandPKL18}
Cl{\'{e}}ment Bertrand, Fr{\'{e}}d{\'{e}}ric Peschanski, Hanna Klaudel, and
  Matthieu Latapy.
\newblock Pattern matching in link streams: Timed-automata with finite memory.
\newblock {\em Sci. Ann. Comput. Sci.}, 28(2):161--198, 2018.
\newblock URL: \url{http://www.info.uaic.ro/bin/Annals/Article?v=XXVIII2\&a=1}.

\bibitem{BJORKLUND2010702}
Henrik Björklund and Thomas Schwentick.
\newblock On notions of regularity for data languages.
\newblock {\em Theoretical Computer Science}, 411(4):702--715, 2010.
\newblock Fundamentals of Computation Theory.
\newblock URL:
  \url{https://www.sciencedirect.com/science/article/pii/S0304397509007518},
  \href {https://doi.org/10.1016/j.tcs.2009.10.009}
  {\path{doi:10.1016/j.tcs.2009.10.009}}.

\bibitem{BojanczykDMSS11}
Mikolaj Bojanczyk, Claire David, Anca Muscholl, Thomas Schwentick, and Luc
  Segoufin.
\newblock Two-variable logic on data words.
\newblock {\em {ACM} Trans. Comput. Log.}, 12(4):27:1--27:26, 2011.
\newblock \href {https://doi.org/10.1145/1970398.1970403}
  {\path{doi:10.1145/1970398.1970403}}.

\bibitem{DBLP:conf/acsd/DeharbeP14}
Aurelien Deharbe and Fr{\'{e}}d{\'{e}}ric Peschanski.
\newblock The omniscient garbage collector: {A} resource analysis framework.
\newblock In {\em 14th International Conference on Application of Concurrency
  to System Design, {ACSD} 2014, Tunis La Marsa, Tunisia, June 23-27, 2014},
  pages 102--111. {IEEE} Computer Society, 2014.
\newblock \href {https://doi.org/10.1109/ACSD.2014.18}
  {\path{doi:10.1109/ACSD.2014.18}}.

\bibitem{Demri09}
St\'{e}phane Demri and Ranko Lazi\'{c}.
\newblock {LTL} with the freeze quantifier and register automata.
\newblock {\em ACM Trans. Comput. Logic}, 10(3), April 2009.
\newblock \href {https://doi.org/10.1145/1507244.1507246}
  {\path{doi:10.1145/1507244.1507246}}.

\bibitem{GT2016}
Radu Grigore and Nikos Tzevelekos.
\newblock History-register automata.
\newblock {\em Log. Methods Comput. Sci.}, 12(1), 2016.
\newblock \href {https://doi.org/10.2168/LMCS-12(1:7)2016}
  {\path{doi:10.2168/LMCS-12(1:7)2016}}.

\bibitem{Grumberg10}
Orna Grumberg, Orna Kupferman, and Sarai Sheinvald.
\newblock Variable automata over infinite alphabets.
\newblock In Adrian{-}Horia Dediu, Henning Fernau, and Carlos
  Mart{\'{\i}}n{-}Vide, editors, {\em Language and Automata Theory and
  Applications, 4th International Conference, {LATA} 2010, Trier, Germany, May
  24-28, 2010. Proceedings}, volume 6031 of {\em Lecture Notes in Computer
  Science}, pages 561--572. Springer, 2010.
\newblock \href {https://doi.org/10.1007/978-3-642-13089-2\_47}
  {\path{doi:10.1007/978-3-642-13089-2\_47}}.

\bibitem{Kaminski1994}
Michael Kaminski and Nissim Francez.
\newblock Finite-memory automata.
\newblock {\em Theor. Comput. Sci.}, 134(2):329--363, 1994.
\newblock \href {https://doi.org/10.1016/0304-3975(94)90242-9}
  {\path{doi:10.1016/0304-3975(94)90242-9}}.

\bibitem{KaminskiZ10}
Michael Kaminski and Daniel Zeitlin.
\newblock Finite-memory automata with non-deterministic reassignment.
\newblock {\em Int. J. Found. Comput. Sci.}, 21(5):741--760, 2010.
\newblock \href {https://doi.org/10.1142/S0129054110007532}
  {\path{doi:10.1142/S0129054110007532}}.

\bibitem{Kara16}
Ahmet Kara.
\newblock {\em Logics on data words: Expressivity, satisfiability, model
  checking}.
\newblock PhD thesis, Technical University of Dortmund, Germany, 2016.
\newblock URL: \url{http://hdl.handle.net/2003/35216}.

\bibitem{MurawskiRT18}
Andrzej~S. Murawski, Steven~J. Ramsay, and Nikos Tzevelekos.
\newblock Polynomial-time equivalence testing for deterministic fresh-register
  automata.
\newblock In Igor Potapov, Paul~G. Spirakis, and James Worrell, editors, {\em
  43rd International Symposium on Mathematical Foundations of Computer Science,
  {MFCS} 2018, August 27-31, 2018, Liverpool, {UK}}, volume 117 of {\em
  LIPIcs}, pages 72:1--72:14. Schloss Dagstuhl - Leibniz-Zentrum f{\"{u}}r
  Informatik, 2018.
\newblock \href {https://doi.org/10.4230/LIPIcs.MFCS.2018.72}
  {\path{doi:10.4230/LIPIcs.MFCS.2018.72}}.

\bibitem{Neven04}
Frank Neven, Thomas Schwentick, and Victor Vianu.
\newblock Finite state machines for strings over infinite alphabets.
\newblock {\em ACM Trans. Comput. Logic}, 5(3):403–435, July 2004.
\newblock \href {https://doi.org/10.1145/1013560.1013562}
  {\path{doi:10.1145/1013560.1013562}}.

\bibitem{SAKAMOTO2000297}
Hiroshi Sakamoto and Daisuke Ikeda.
\newblock Intractability of decision problems for finite-memory automata.
\newblock {\em Theoretical Computer Science}, 231(2):297 -- 308, 2000.
\newblock URL:
  \url{http://www.sciencedirect.com/science/article/pii/S030439759900105X},
  \href {https://doi.org/10.1016/S0304-3975(99)00105-X}
  {\path{doi:10.1016/S0304-3975(99)00105-X}}.

\bibitem{Tzevelekos11}
Nikos Tzevelekos.
\newblock Fresh-register automata.
\newblock In Thomas Ball and Mooly Sagiv, editors, {\em Proceedings of the 38th
  {ACM} {SIGPLAN-SIGACT} Symposium on Principles of Programming Languages,
  {POPL} 2011, Austin, TX, USA, January 26-28, 2011}, pages 295--306. {ACM},
  2011.
\newblock \href {https://doi.org/10.1145/1926385.1926420}
  {\path{doi:10.1145/1926385.1926420}}.

\end{thebibliography}
\newpage
\appendix
\section{Simulating the membership for \lama in {$\nu$}{nu}-automata -- proofs}
\label{app:simul-member-lama-nu}
The proof of the main theorem relies on some lemmata, showing that the encoding is correct.
\begin{restatable}[Non-observable transitions]{lemma}{nonobstrans}
\label{lem:nonobstrans}
Let $A$ be an $n$-\lama. $(q,M) \xrightarrow[A]{\varepsilon} (q', M')$ if and only if $(q,\enco{M}_\xi) \xrightarrow[\enco{A}_{\xi}]{\varepsilon} (q', \enco{M'}_\xi)$.
\end{restatable}

\begin{proof} 
For both sides the proof follows by the construction of the \nuauto, observing that non-observable transitions have the same semantics in both families of automata.
\end{proof}

\begin{lemma}[Observable transitions]\label{lem:obstrans}
Let $A$ be an $n$-\lama. $(q,M) \xrightarrow[A]{u} (q', M_1)$ if and only if $$(q,\enco{M}_\xi) \xRightarrow[\enco{A}_{\xi}]{\xi^n(u)} (q', \enco{M_1}_\xi)$$.
\end{lemma}
\begin{proof}
Let $(q,M) \xrightarrow[A]{u} (q', M_1)$ then by definition we know that $\delta = (q, \alpha, q') \in \setobs$ and by construction of $\enco{A}_{\xi}$, there is a path of length $n$ of observable transitions from $q$ to $q'$.
Hence, let $\xi^n(u)= u^1 \dots u^n$, we need to prove that $(q,\enco{M}) \xRightarrow[\enco{A}_{\xi}]{\xi^n(u)} (q', M_2)$ with $M_2 = \enco{M_1}_\xi$.
Each of these subsequent transitions 
considers the changes to a variable of the corresponding layer. By induction of the length of this path $n$. The base case $n = 1$ is trivial as a $1$-\lama is also a \nuauto. For $n > 1$, we know by inductive hypothesis that 
$(q,\enco{M}_\xi) \xRightarrow[\enco{A}_{\xi}]{u^1 \dots u^{n-1}} (q_{\delta}^{n-1}, M_3 )$
where for all $l<n, M_3(\proj{\alpha(l)}{1}) = \enco{M_1}_\xi(\proj{\alpha(l)}{1})$ and $ M_3(\proj{\alpha(n)}{1}) = \enco{M}_\xi(\proj{\alpha(n)}{1})$. 
Now for the last step in the derivation, we have three possible cases, let $\proj{\alpha(n)}{1} = v^n$:
\begin{description}
\item[$\proj{\alpha(n)}{2} = \reads$:]  by hypothesis we have $u \in M(v)$ and by construction $u^n \in \enco{M}_\xi(v)$ hence the transition can be executed  in $\enco{A}_{\xi}$ leading to $(q', \enco{M_1}_\xi)$.

\item[$\proj{\alpha(n)}{2} = \writes$:] by hypothesis  $\forall v_1^n \in V, u \not\in M(v_1^n)$, and $M'(v^n) = M(v^n)\cup\{u\}$. 
By construction, we know that $\forall v_1^n \in V', u^n \not\in \enco{M}_\xi(v_1^n)$. Moreover by definition of $\xi^n(u)$, $u^n \not\in \enco{M}_\xi(v_2)$ for all $v_2 \in V'$. Hence injectivity is satisfied and $u^n$ is added to $\enco{M}_\xi(v)$ as expected leading to $(q', \enco{M_1}_\xi)$.

\item[$\alpha(n) = \ignore$:] the letter is consumed without examining any variable, leading  to $(q', \enco{M_1}_\xi)$.
\end{description}

For the other direction, let $\xi^n(u)= u^1 \dots u^n$ we know  that $(q,\enco{M}_\xi) \xRightarrow[\enco{A}_{\xi}]{\xi^n(u)} (q', M_2)$, hence by construction  $\delta = (q, \alpha, q') \in \setobs$. We need to show that $(q,M) \xrightarrow[A]{u} (q', M_1)$  and $M_2 = \enco{M_1}_\xi$.

Similarly as before the proof proceeds by induction on the length of the path from $q$ to $q'$ in $\enco{A}_{\xi}$. For $n > 1$, we have that 
$(q,\enco{M}) \xRightarrow[\enco{A}_{\xi}]{u^1 \dots u^{n-1}} (q_{\delta}^{n-1}, M_3 )$ and by inductive hypothesis for all $l<n, M_3(\proj{\alpha(l)}{1}) = \enco{M_1}_\xi(\proj{\alpha(l)}{1})$ and $ M_3(\proj{\alpha(n)}{1}) = \enco{M}_\xi(\proj{\alpha(n)}{1})$. 
Now for the last step in the derivation, we have three possible cases, let $\proj{\alpha(n)}{1} = v^n$:
\begin{description}
\item[$\proj{\alpha(n)}{2} = \reads$:]  by hypothesis we have $u^n \in \enco{M}_\xi(v^n) = M_3(v^n)$ thus by definition of the encoding of a memory context $\enco{\cdot}_\xi$ $u \in M(v^n)$ hence the transition can be executed  in $A$ leading to $(q', M_1)$.

\item[$\proj{\alpha(n)}{2} = \writes$:] by hypothesis $u^n \not\in \enco{M}_\xi(v_1)$ 
for all $v_1 \in V'$ this entails  that $\forall v_1^n \in V, u \not\in M(v_1^n)$ thus leading to $(q', M_1)$.

\item[$\alpha(n) = \ignore$:] the letter is consumed without examining any variable, leading again to $(q', M_1)$.
\end{description}

\end{proof}

\lamatonu*
\begin{proof}
The proof follows (for both sides) by induction on the length of the derivation of $w$ and $\xi^n(w)$ respectively using previous Lemmas~\ref{lem:nonobstrans} and \ref{lem:obstrans}.
\end{proof}

\section{Simulating the membership for \hra in \lama~-- proofs}
\label{app:simul-member-hra-lama}
Similarly as before, we first need to introduce some auxiliary definitions and lemmata.

\begin{definition}
    Let $M$ be a memory context, $H = h_1 \dots h_n$ a set of histories of an \hra $A$ and $w \in \alphab^*$. Then a memory context $M'$ for the set $V= h^1, \dots, h^n, \omega^1,\dots, \omega^n$ of variables of $\enco{A}_\zeta$ is \emph{well-formed with respect to $w$ and $M$} if it satisfies the following constraints:
     For all $l \in [1,n]$, if  $u\in M(h_l)$ then 
  \begin{enumerate}
         \item if $u$ appears in $w$ and $u_j$ is the last occurrence of $u$ in $w$ then $\proj{\zeta_{i_{u_j}}(u)}{2} \in M'(h^l)$
         \item if $u$ does not appear in $w$ then $\proj{\zeta_{1}(u)}{1} \in M'(h^l)$ (i.e., $u_0 \in  M'(h^l)$)
     \end{enumerate}
For all $l \in [1,n]$, if  $u \notin M(h_l)$ then 
  \begin{enumerate}
         \item if $u$ appears in $w$ and $u_j$ is the last occurrence of $u$ in $w$ then $\proj{\zeta_{i_{u_j}}(u)}{2} \notin M'(h^l)$
         \item if $u$ does not appear in $w$ then $\proj{\zeta_{1}(u)}{1} \notin M'(h^l)$ 
     \end{enumerate}

\end{definition}

\begin{lemma} \label{lem:encodingHRAtoLAMA}
    Let $A$ be an \hra and $w=u_1\dots u_m \in \alphab^*$. $(q_0,M_0) \xRightarrow[A]{w} (q,M)$ if and only if $(q_0,M_0') \xRightarrow[\enco{A}_\zeta ]{\zeta(w)} (q,M')$ with $M'$ well-formed with respect to $w$ and $M$.
\end{lemma}
\begin{proof}
~
\begin{description}
    \item[$\Rightarrow$]

The proof proceeds by induction on the length of the derivation of $w$ with a case analysis on the last step of the derivation (with $w = w' \cdot x$):
$$(q_0,M_0) \xRightarrow[A ]{w'} (q_p,M_p) \xrightarrow[A ]{x} (q,M)$$

\begin{description}
\item[Non-observable transitions:] Let $(q_p,M_p) \xrightarrow[A]{\varepsilon} (q,M)$ then  $x = \varepsilon$, $\delta=(q_p,H_\emptyset,q) \in \setsilent$. 
By inductive hypothesis we know that $(q_0,M_0') \xRightarrow[\enco{A}_\zeta ]{\zeta(w)} (q_p,M_p')$ and $M_p'$ is well-formed with respect to $w$ and $M_p$. By construction $\delta'=(q_p, \rset ,q) \in \setsilent'$ with $\rset=\{h^i\mid h_i \in H_\emptyset\}$ hence $(q_p,M_p') \xrightarrow[\enco{A}_\zeta ]{\varepsilon} (q,M')$
where all variables in $\rset$ have been reset. It immediately follows that $M'$ is well-formed with respect to to $w$ and $M$. 
    
\item[Observable transitions:] Let 
$(q_p,M_p) \xrightarrow[A]{u_m} (q,M)$ then $x = u_m$, $\delta = (q_p,H_r,H_w,q) \in \setobs $.
By inductive hypothesis we know that $(q_0,M_0') \xRightarrow[\enco{A}_\zeta ]{\zeta(u_1 \dots u_{m-1})} (q_p,M_p')$ and $M_p'$ is well-formed with respect to $u_1 \dots u_{m-1}$ and $M_p$.
Moreover by construction there are two transitions $(q_p, \alpha_{H_r}, q_\delta)$ and $(q_\delta, \alpha_{H_w} ,q)$ in $\setobs'$. Hence
$(q_p,M_p') \xrightarrow[\enco{A}_\zeta ]{\proj{\zeta_{i_{u_m}}(u_m)}{1}} (q_{\delta},M'') \xrightarrow[\enco{A}_\zeta]{\proj{\zeta_{i_{u_m}}(u_m)}{2}} (q, M')$.

Transition $(q_p,M_p') \xrightarrow[\enco{A}_\zeta ]{\proj{\zeta_{i_{u_m}}(u_m)}{1} } (q_{\delta},M'')$ can be executed as, since $M_p'$ is well formed with respect to $u_1 \dots u_{m-1}$ and $M_p$, and  $u_m\in M_p(h_i)$, $\forall h_j \in H_r$ then $\proj{\zeta_{i_{u_m}}(u_m)}{1} \in M_p'(h^j)$.
Moreover since $u_m \notin M_p(h_j)$ for all $h_j \notin H_r$ then $\proj{\zeta_{i_{u_m}}(u_m)}{1} \notin M_p'(h^j)$ then $\proj{\zeta_{i_{u_m}}(u_m)}{1}$ can be added to $\omega^j$. Finally, transition $(q_{\delta},M'') \xrightarrow[\enco{A}_\zeta]{\proj{\zeta_{i_{u_m}}(u_m)}{2}} (q, M')$ is always enabled as $\proj{\zeta_{i_{u_m}}(u_m)}{2}$ is a fresh letter.

As a last step, we need to show that $M'$ is well-formed with respect to $w$ and $M$.
Now take $u \in M(h_i)$ there are three cases:
\begin{enumerate}
\item $u$ does not appear in $w$ then $u \in M_0(h_i)$ and hence by construction $\zeta_0(u) \in M_0'(h^i)$.
 \item $u =u_j \neq u_m$, $j \in [1,m-1]$ then $u \in M_p(h_i)$ and  since $M_p'$ is well-formed with respect to to $u_1 \dots u_{m-1}$ and $M_p$,  then $\proj{\zeta_{i_{u_j}}(u)}{2} \in M'(h^i)$ with $u_j$ being the last occurrence of $u$ in $w$.  
    \item $u=u_m$ then this means that $h_i \in H_w$. Hence for what has been shown above $\proj{\zeta_{i_{u_m}}(u_m)}{2} \in M'(h^i)$.  
\end{enumerate}
Moreover for all $u \notin M(h_i)$ there are two cases
\begin{enumerate}
    \item $u$ not occurs in $w$ then $u \notin M_0(h_i)$ and by construction $\proj{\zeta_1(u)}{1} \notin M_0'(h^i)$
     \item $u$ occurs in $w$, let $u_j$ be the last occurrence of $u$ in $w$ then in the execution $(q_0,M_0) \xRightarrow[A]{w} (q,M)$ there is a transition  $(q_1,M_1) \xrightarrow[A]{u_j} (q_2,M_2)$ with $\delta' = (q_1,H_r,H_w,q_2) \in \setobs $ and $h_i \notin H_w$ and hence by construction $\proj{\zeta_{i_{u_j}}(u)}{2} \notin M'_2(h^i)$  and hence $u_j \notin M'(h^i)$. 
\end{enumerate}

\end{description}
 \item[$\Leftarrow$] 
Suppose $(q_0,M_0') \xRightarrow[\enco{A}_\zeta ]{\zeta(w)} (q,M')$ with $M'$ well-formed with respect to $w$ and some $M$.
We say that a step of the derivation is a sequence of transitions $(q_1,M_1) \xRightarrow[\enco{A}_\zeta ]{x} (q_2,M_2)$ with $q_1, q_2 \in Q$, $x = \zeta_{i_{u_j}}(u_j)$ or $x= \varepsilon$ and all intermediate states $q' \notin Q$.
The proof proceeds by induction on the length of the derivation of $\zeta(w)$ with a case analysis on the last step of the derivation:

\begin{description}
\item[Non-observable transitions:] 
Let $(q_p,M_p') \xrightarrow[\enco{A}_\zeta ]{\varepsilon} (q,M')$ with $M'$  well-formed with respect to $w$ and $M$, then $\delta'=(q_p, \rset ,q) \in \setsilent'$.  
Now by construction there exists $\delta=(q_p,H_\emptyset,q) \in \setsilent$ such that $\rset=\{h^i\mid h_i \in H_\emptyset\}$. By inductive hypothesis we know that $(q_0,M_0) \xRightarrow[A]{w} (q_p,M_p)$ and $M_p'$ well-formed with respect to $w$ and $M_p$. Now the transition can be taken in $A$ and all variables in $H_{\emptyset}$ are emptied maintaining the well-formedness of $M'$.

\item[Observable transitions:] Let 
$(q_p,M_p') \xRightarrow[\enco{A}_\zeta]{\zeta_{i_{u_m}}(u_m)} (q,M')$ 
then by construction there exists a transition $\delta=(q_p, H_r, H_w, q) \in \setobs$ such that there
are two transitions $(q_p, \alpha_{H_r}, q_\delta)$ and $(q_\delta, \alpha_{H_w} ,q)$ in $\setobs'$. 
By inductive hypothesis we know that $(q_0,M_0) \xRightarrow[A]{u_1 \dots u_{m-1}} (q_p,M_p)$ and $M_p'$ is well-formed with respect to $u_1 \dots u_{m-1}$ and $M_p$. Now $\delta$ can be executed because of well formedness of $M_p'$ if $\proj{\zeta_{i_{u_m}}(u_m)}{1} \in h^r$ with $h_r \in H_r$ then $u_m \in M_p(h_r)$, moreover we know that $\proj{\zeta_{i_{u_m}}(u_m)}{1} \notin M_p'(h^i)$ for all $h_i \not\in H_r$ hence $u_m \notin M_p(h_i)$. 
Then, transition $(q_p,M_p) \xrightarrow[A]{u_m} (q, M)$ can be executed and $u_m$ is added to variables $h \in H_w$, keeping the well formedness of $M'$ with respect to to $w$ and $M$.
\end{description}
\end{description}
    
\end{proof}

And finally, we can conclude:

\hratonu*
\begin{proof}
 $w \in \lang{A}$ if and only if $(q_0,M_0) \xRightarrow[A]{w} (q,M)$ with $q \in F$ if and only if by Lemma~\ref{lem:encodingHRAtoLAMA} we have that $(q_0,M_0') \xRightarrow[\enco{A}_\zeta ]{\zeta(w)} (q,M')$ if and only if by construction $q \in F'$ if and only if $\zeta(w) \in \lang{\enco{A}_\zeta}$.
\end{proof}

\section{Complexity of membership -- proofs}
\label{app:complexity-member-lama}
\lamanp*
\begin{proof}
Let LMP be the short for the membership problem for \lama.
{\bf NP-hardness:} we show that 3SAT can be reduced to LMP. 
Let $\exists X_1, \ldots \exists X_n C_1 \wedge \ldots \wedge  C_m$, where each $C_i$ is a clause of exactly $3$ literals (i.e.,   positive or negative versions of variables 
$X_j$). To encode 3SAT in \lama we consider for each $X_i$ two \lama variables $X_i$ and $\overline{X_i}$ belonging to layer $i$, and an initialisation gadget composed of two states and two transitions depicted in Fig. ~\ref{fig:init-xi}.
For each clause $C_j$ we consider a clause gadget composed of two states and three transitions as depicted in Fig. ~\ref{fig:clause-3SAT}.
In order to construct the $n$-\lama $A$ encoding the instance of 3SAT we connect $n$ initialisation gadgets (one for each $X_i$) and then the $m$ clause gadgets. The initial state of the first variable is the initial state of $A$ and the exit state of the last clause gadget is the unique final state of $A$. The input word is $w=(1)^n (1)^m$, where the $n$ occurrences of letter $1$ are used to non-deterministically initialise variable $X_i$ or $\overline{X_i}$. The next $m$ occurrences of letter $1$ are used to accept at least one literal in each clause. It is easy to see that $w$ is accepted by $A$ starting with an empty memory context if and only if there is a solution to the 3SAT instance.

{\bf NP-membership:} 
A certificate that a \lama $A=(Q, q_0, F, \Delta, V, n, M_0)$ accepts $w$ is a sequence $\sigma=s_0t_1s_1t_2\ldots s_m$ where all $s_i$ are configurations of the form $(q,M)$ with $q\in Q$ a state of $A$ and $M$ a memory context, $s_0=(q_0,M_0)$, and each $t_i$ is a transition in $\Delta$. The certifier executes $\sigma$ and checks if (1)  $s_m=(q',M')$ is a final configuration, i.e., if $q'\in F$ and (2) if $m\leq (|w| + |w|(|Q|\cdot |V|))$.
This upper bound comes from the observation that: 
\begin{itemize}
    \item each letter in $w$ needs exactly one transition;
    \item  between two consecutive observable transitions in $\sigma$ there could be at most $O(|V|\cdot |Q|)$  reset transitions,  as resetting the same variable twice has no effect and reaching the transition resetting a given variable can take at most $|Q|-1$ steps. 
\end{itemize}

This verification can be done in polynomial time since each observable transition takes $O(|V|\cdot|w|)$ steps and each reset transition  $O(1)$.

\end{proof}

\section{Non-emptiness for {$\nu$}{nu}-automata is in PSPACE -- proofs}
\label{app:empty-pspace}

\simplisubset*

\begin{proof}
Let $w\in\lang{\simpli(A)}$. Then there exists a path $(q_0,M'_0) \xRightarrow[\simpli(A)]{w} (q_f,M'_f)$ with $q_f \in F$ and intermediate configurations $(q_i,M'_i)_{i\geq 0}$.
We show by induction that  $w \in \lang{A}$ and the path accepting it is  a sequence of configurations $(q_i,M_i)_{i\geq 0}$ with $M'_i$ being the projection of $M_i$ over the letters in $K$. 

\begin{description}
    \item[Base case:] both automata start from the same initial state $q_0$. By definition, $\forall v\in V: k_v\in M_0(v)$ if and only if $M'_0(v)\neq \emptyset$. So, the property is true in the initial configuration. 

    \item[Inductive step:]  We examine the  transition  $(q_i,M'_i) \xrightarrow[\simpli(A)]{}(q_{i+1}, M'_{i+1})$. By construction, the same transition must exist in $A$. Now, depending on the transition,  consider the resulting memory context:
\begin{description}
    \item[reset:] the memory context of the variables in both automata are emptied, thus satisfying the property;
    \item[read:] if it is a read of letter $u$ in $v$, the transition is fired in $\simpli(A)$ so $u=k_v$. By inductive hypothesis, in the memory context of $A$ we have $k_v \in M_i(v)$;
    \item[write:] if it is a write of letter $u$ in $v$:
    \begin{itemize}
        \item if $k_v \not\in M'_i(v)$, then $M_i(v) = \emptyset$ and it means $u = k_v$ as it is always the first letter written in an empty variable in $\simpli(A)$. Thus $M_{i+1}(v)=M'_{i+1}(v)$.
        \item if $k_v \in M'_i(v)$ then letter $u$ is a token, $u = t_i \in T$, and is written in $v$ for $A$, $t_i \in M_{i+1}(v)$ but not for $\simpli(A)$. The projection of $M_{i+1}(v)$ to $K$ will still contain $k_v$, thus the invariant is satisfied.
    \end{itemize}
\end{description}
    
\end{description}

\end{proof}

\begin{definition} [FSM]
Let $A=(Q, q_0, F, \Delta, V, M_0)$ be a \nuauto and \simpli(A) its canonical version, we define an FSM $\fsm(A)=(K, S, s_0', E, \delta)$, where
\begin{itemize}
    \item $K = \{ k_v \mid v\in V \}$ is the alphabet;
    \item $S = Q \times 2^K$ is the set of  states, where $(q,m) \in S$ abstracts  configuration $(q,M)$ of $A$ where $k_v \in m \text{ if and only if } k_v\in M(v)$;
    \item $s_0 = (q_0, m_0)$, with $m_0=\{ k_v \mid M_0(v)\neq \emptyset \}$ is the initial state;
    \item $E=F\times 2^K$ is the set of  final states;
    \item $\delta \subseteq S \times (K\cup \{\epsilon\}) \times S$ is the set of  transitions with 
    \begin{itemize}
        \item $((q,m),k_v,(q',m)) \in \delta$ if and only if $((q,M),(v,\reads),(q',M)) \in \Delta$ and $k_v \in m$ (reads a key);
        \item $((q,m),k_v,(q',m')) \in \delta$ if and only if $((q,M),(v,\writes),(q',M')) \in \Delta$, $k_v \not \in m$ and $m' = m \cup \{k_v\}$ (writes a key);
        \item $((q,m),\epsilon,(q',m)) \in \delta$ if and only if $((q,M),(v,\writes),(q',M')) \in \Delta$ and $k_v\in m$ (writes a token);
        \item $((q,m),\epsilon,(q',m')) \in \delta$ if and only if $((q,M),\rset,(q',M')) \in \Delta$ and $m' = m \setminus \rset$ (resets variables in $\rset$.
    \end{itemize}
\end{itemize}

\end{definition}

\section{\lama should not  be in PSPACE}
\label{app:not-in-pspace}

The \lama presented in  Fig.  \ref{fig:not-in-pspace} accepts words of length in $O(2^{2^n})$, when its initial memory context is  empty. 
We can decompose this automaton in three parts:
\begin{itemize}
    \item The first part consists of the state $q_0$ with a self-loop allowing to write as many different letters as needed in $X^1$. This is the only transition actually adding new letters to $X^1$. The same $X^1$ is read to enable the final transition to $q_f$. State $q_0$ can be left through a $\varepsilon$-transition to $q_1$.
    \item The second part includes the states from $q_1$ to $q_3$ will erase at least half of the letters stored in $X^1$. To do so, the two looping transitions between $q_1$ and $q_2$ will split in half the letters of $X^1$ between $Y^2$ and $Z^2$. Then the transition from $q_1$ to $q_3$ reset $X^1$, and the self-loop on  $q_3$ allows to store back in $X^1$ the letters that were copied in $Y^2$. Finally, the outgoing transition of $q_3$ resets $Y^2$ and $Z^2$, definitely erasing the letters from $Z^2$. The whole process is non-deterministic, however, it is actually required to copy enough letters through this process so  that the reset from $q_3 \rightarrow q_4$ will erase every letter in $X^1$, which is required for the final transition.
    \item The third part contains all states from $q_4$ to $q_m$. It is a repetition of the process described so far to ensure an exponential amount of iterations over the second part. In state $q_{m-2}$ (and respectively for all the previous states), there are two outgoing transitions: one that requires to read in $A^n$ and the other one that writes in (initialize) it. The reading transition is the one leading to the final state but it cannot be crossed as long as $A^n$ was not initialized through the other transition. The writing transition leads to the state $q_{m-1}$ and eventually back to $q_1$, while resetting every previous $A^i$ for $i < n$. As these variables are empty, all their initializations have to be done again.
\end{itemize}
The third part  iterates back to the states of the second part $2^{n-2}$ times. Every time we pass through the second part, we lose at least half of the letters stored in $X^1$.  Hence to be able to cross the last transitions, it is required that during the first part, at least $2^{2^{n-2}}$ letters are stored in $X^1$.

\begin{figure}
    \centering
\begin{tikzpicture}[node distance=2cm]
    \node[state] (q1) at (0,0) {$q_1$};
    \node[state, left of = q1, node distance = 1.5cm] (q0) {$q_0$};
    \node [left of = q0, node distance = 1cm] (start) {};
    \draw [->] (start)--(q0);
    \node[state, below of = q1] (q2) {$q_2$};
    \node[state, right of = q1, node distance = 2.5cm] (q3) {$q_3$};
    \node[state, right of = q3, node distance = 2.5cm] (q4) {$q_4$};
    \node[state, above of = q4] (q5) {$q_5$};
    \node[state, right of = q4] (q6) {$q_6$};
    \node[state, above of = q6] (q7) {$q_7$};
    \node[right of = q6] (dots) {$\cdots$};
    \node[state, right of = dots] (qm2) {$q_{m-2}$};
    \node[state, above of = qm2] (qm1) {$q_{m-1}$};
    \node[state, below of =qm2] (qm) {$q_m$};
    \node[state, left of =qm, double] (qf) {$q_f$};

    \draw [->] (q0) to[out= 50, in=130, looseness=5]node[above]{$(X^1,\writes)$} (q0);
    \draw [->] (q0) to node[above]{$\varepsilon$} (q1);
    \draw [->] (q1) to[out= -45, in=45]node[right]{$\begin{array}{l}(X^1,\reads) \\(Y^2,\writes)\end{array}$} (q2);
    \draw [->] (q2) to[out= 135, in=-135]node[left]{$\begin{array}{l}(X^1,\reads) \\(Z^2,\writes)\end{array}$} (q1);
    \draw [->] (q1) to node[above]{$\rset(X^1)$} (q3);
    \draw [->] (q3) to[out= 50, in=130, looseness=5]node[above]{$\begin{array}{l}(X^1,\writes) \\(Y^2,\reads)\end{array}$} (q3);
    \draw [->] (q3) to node[below]{$\rset(Y^2, Z^2)$} (q4);
    \draw [->] (q4) to node[left]{$(A^3, \writes)$} (q5);
    \draw [->] (q4) to node[above]{$(A^3, \reads)$} (q6);
    \draw [->] (q6) to node[above]{$(A^4, \reads)$} (dots);
    \draw [->] (q6) to node[left]{$(A^4, \writes)$} (q7);
    \draw [->] (dots) to node[above]{$(A^n, \writes)$} (qm2);
    \draw [->] (qm2) to node[left]{$(A^{n-1}, \writes)$}(qm1);
    \draw [->] (qm2) to node[left]{$(A^n, \reads)$}(qm);
    \draw [->] (qm) to node[below]{$(X^1, \reads)$}(qf);

    \draw [->, rounded corners] (q5) --node[above]{$\varepsilon$} (0, 2)-- (q1);
    \draw [->, rounded corners] (q7) to[out=160, in= 10]node[above]{$\rset(A^3)$} (0, 2.3)-- (q1);
    \draw [->, rounded corners] (qm1) to[out = 155,in = 15]node[above]{$\rset(\{A^l \mid l \in [3,m-1]\})$} (0, 2.6)--(q1);
    
\end{tikzpicture}

    \caption{Example of a \lama with words whose length is in $O(2^{2^n})$.}
    \label{fig:not-in-pspace}
\end{figure}
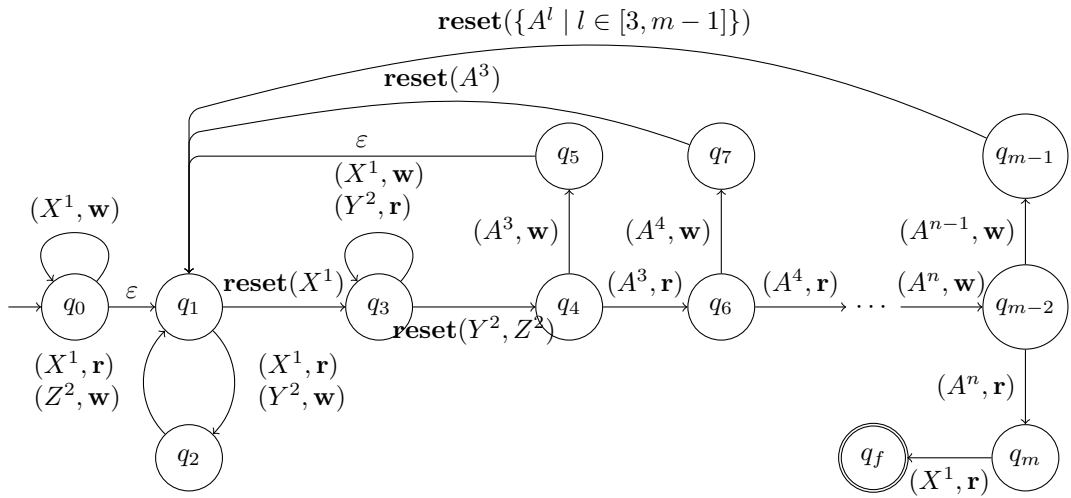
\end{document}